\documentclass[11pt]{article}

\usepackage[margin=1in]{geometry}
\usepackage{amsmath,amssymb,amsthm}
\usepackage{bm}
\usepackage{mathtools}
\usepackage{graphicx}
\usepackage{booktabs}
\usepackage{tabularx}
\usepackage{algorithm}
\usepackage{algpseudocode}
\usepackage[colorlinks=true,linkcolor=blue,citecolor=blue,urlcolor=blue]{hyperref}
\usepackage[round]{natbib}
\usepackage{enumitem}

\newtheorem{theorem}{Theorem}[section]
\newtheorem{proposition}[theorem]{Proposition}
\newtheorem{lemma}[theorem]{Lemma}

\newtheorem{assumption}{Assumption}

\newtheorem{remark}[theorem]{Remark}

\DeclareMathOperator*{\argmin}{arg\,min}
\DeclareMathOperator{\conv}{conv}

\title{\textbf{Risk-Set Transported Synthetic Control with Difference-in-Differences Adjustment \\ under Staggered Treatment Adoption}}

\author{
  Mojtaba Eslami \\[4pt] 
  \small University of Calgary, \\  
  \small \texttt{mojtaba.eslami@alumni.ucalgary.ca}
}

\date{June 01, 2026}

\begin{document}

\maketitle
\vspace{-2.5em}
\begin{center}
\textit{Working paper. Comments welcome.}
\end{center}
\vspace{1em}

\begin{abstract}
\noindent
In staggered treatment-adoption designs, later-treated units can serve as controls for an earlier-treated cohort only until their own treatment begins. The admissible donor set therefore contracts with event time. Fixing the donor pool at the longest evaluation horizon discards temporarily eligible donors, while independently re-estimating synthetic-control weights at every horizon can make the counterfactual sensitive to mechanical changes in donor composition. We propose \emph{Risk-Set Transported Synthetic Control with Difference-in-Differences Adjustment} (RT-SC-DiD). At each cohort and event-time horizon, the estimator fits weights on the currently untreated donor set while regularizing them toward a transported reference that reallocates exiting donors' weight toward similar surviving donors. A DiD baseline correction removes persistent level differences. We characterize the distortion induced by reoptimization, derive the loading change caused by naive deletion and renormalization, and provide a conditional recursive bound for error propagation across horizons under an explicitly stated (not derived) regularity condition on the transport map. We also introduce donor-support diagnostics and a donor-only placebo procedure for selecting the transport penalty. In an 80-replication pilot method comparison, together with a separate 40-replication-per-value transport-strength sensitivity analysis, intermediate transport regularization lowers average RMSE relative to independent horizon-specific estimation and to strong anchoring, consistent with the bias--variance motivation of the method rather than a result we prove. The method is intended for settings in which later-treated units contain useful short-horizon information but the admissible donor set contracts materially over the evaluation window. To the best of our knowledge, existing staggered synthetic-control and SDiD methods do not explicitly regularize within-cohort weight sequences toward transported references as the admissible donor set contracts across event-time horizons.
\end{abstract}

\noindent\textbf{Keywords:} synthetic control, difference-in-differences, staggered adoption, panel data, causal inference, not-yet-treated units, transport regularization

\section{Introduction}
\label{sec:intro}

A large share of modern policy evaluation relies on panel data in which units adopt a treatment at different calendar dates rather than simultaneously. Environmental regulations, hospital technologies, transportation policies, and educational reforms are all frequently rolled out sequentially across jurisdictions, facilities, or firms. This staggered structure has motivated a substantial literature on difference-in-differences (DiD) estimation with variation in treatment timing, most prominently the group-time average treatment effect framework of \citet{callaway2021difference}, alongside related estimators by \citet{sun2021estimating}, \citet{de2020two}, \citet{borusyak2024revisiting}, and \citet{goodman2021difference}. A parallel literature extends synthetic control methods \citep{abadie2003economic,abadie2010synthetic} to panels with multiple treated units and staggered timing, including the partially pooled synthetic control of \citet{ben2021synthetic}, the synthetic difference-in-differences (SDiD) estimator of \citet{arkhangelsky2021synthetic}, and matrix-completion approaches to counterfactual imputation \citep{xu2017generalized,athey2021matrix}.

A recurring feature of staggered designs is that units treated later than a given cohort can, prior to their own treatment date, serve as informative controls for that cohort's counterfactual path. This idea underlies the ``not-yet-treated'' comparison group in \citet{callaway2021difference} and is explicitly acknowledged in the staggered synthetic control literature \citep{ben2021synthetic}: the set of units eligible to serve as donors for cohort $g$ at event time $h$ depends on $h$, because a donor whose own treatment begins before period $g+h$ can no longer be considered untreated at that date. In practice, however, existing implementations typically resolve this dependence by restricting attention to donors that remain untreated through the \emph{longest} horizon under study. This produces a synthetic counterfactual that is internally consistent across horizons, but it discards potentially valuable short-horizon information from units that are only temporarily eligible, and it can degrade pre-treatment fit when few units remain untreated through the full evaluation window.

The alternative of re-estimating synthetic weights independently at every horizon, using the largest currently valid donor set, avoids discarding information but introduces a new problem: because the donor pool itself changes mechanically as units become treated, the estimated counterfactual path can exhibit jumps that reflect \emph{donor composition change} rather than genuine dynamics in the treatment effect. Distinguishing the two is essential for credible event-study estimation, since a central empirical use of the event-time profile $\{\widehat\tau_{g,0}, \widehat\tau_{g,1}, \dots\}$ is precisely to characterize how treatment effects evolve.

This paper develops an estimator, \emph{Risk-Set Transported Synthetic Control with DiD Adjustment} (RT-SC-DiD), designed to use the full time-varying donor risk set while explicitly controlling for the instability that a shrinking risk set can introduce. The estimator retains a horizon-specific synthetic control fit to the largest currently valid risk set, but regularizes the horizon-$h$ weights toward a \emph{transported} version of the horizon-$(h-1)$ weights, in which the weight held by donors who exit the risk set between $h-1$ and $h$ is redistributed to the remaining donors in proportion to pre-treatment similarity. A difference-in-differences residualization step removes persistent level gaps between the treated cohort and the weighted donor path, and a partially pooled objective permits the procedure to share information across cohorts evaluated at the same horizon when cohort-specific donor pools are thin.

We view this problem as the central object of the paper. Terminology matters here: donors do not exit the risk set because of missing data or sample selection on outcomes -- their exit time $G_i$ is a deterministic function of their own treatment date and is known in advance for units already treated. We therefore avoid the phrase ``endogenous donor attrition,'' which more naturally evokes a missing-data or sample-selection problem, and instead describe the phenomenon as an \emph{informative contraction of the admissible donor set}: the set of units eligible to serve as controls shrinks deterministically with $h$, and whether the units that exit are systematically different from those that remain (in a sense made precise by Assumption~\ref{ass:attrition}) is a substantive question that must be checked, not assumed away by the fact that exit itself is mechanical. We use ``donor attrition'' informally throughout as shorthand for this contraction, not for missingness.

We do not claim to be the first to use not-yet-treated units as synthetic control donors, the first to estimate dynamic treatment effects under staggered adoption, or the first to combine synthetic weighting with a DiD-style level adjustment; each of these ideas has clear antecedents, discussed in Section~\ref{sec:related}. What we believe is new is a transport-regularized sequence of horizon-specific donor weights, together with a proved bound relating transport distance to reweighting-induced counterfactual distortion, a corrected account of the bias in the naive alternative (deletion-and-renormalization), and a recursive bound on error propagation across horizons induced by the transport mechanism itself -- together with a diagnostic toolkit built around these quantities. We are explicit throughout about which results are proved, which are definitions or accounting identities, and which are proposed but not yet fully executed: Section~\ref{sec:simulation-illustration} reports a small illustrative pilot simulation run for this draft, but the fuller simulation design of Section~\ref{sec:simulation} and the empirical application of Section~\ref{sec:application} remain proposed rather than completed.

\paragraph{Roadmap.} Section~\ref{sec:related} situates the proposal relative to existing staggered DiD and synthetic control methods. Section~\ref{sec:setup} defines notation and estimands. Section~\ref{sec:riskset} formalizes horizon-dependent donor risk sets and the donor-attrition problem. Section~\ref{sec:estimator} develops the transport operator and the RT-SC-DiD estimator, including DiD residualization and partial pooling across cohorts. Section~\ref{sec:identification} states identifying assumptions under an interactive fixed-effects model. Section~\ref{sec:theory} derives an error decomposition and states the main theoretical propositions. Section~\ref{sec:tuning} discusses tuning-parameter selection via donor-only placebo validation. Section~\ref{sec:inference} discusses inference. Section~\ref{sec:diagnostics} introduces donor-support diagnostics. Section~\ref{sec:simulation} lays out a Monte Carlo design. Section~\ref{sec:application} discusses criteria for an empirical application and a decomposition of what the method adds relative to fixed-pool alternatives. Section~\ref{sec:falsification} describes falsification and sensitivity checks. Section~\ref{sec:discussion} discusses limitations and scope, and Section~\ref{sec:conclusion} concludes.

\section{Related literature}
\label{sec:related}

\subsection{Group-time treatment effects and not-yet-treated comparisons}
\citet{callaway2021difference} define the group-time average treatment effect $ATT(g,t)$ and estimate it using either a never-treated comparison group or a not-yet-treated comparison group, the latter allowing units that will eventually be treated to serve as controls prior to their own adoption date. This is the origin of the risk-set idea we build on: the admissible comparison group for cohort $g$ at calendar time $t$ is itself a function of $t$. Our contribution relative to this literature is not the use of not-yet-treated controls \emph{per se}, but the replacement of a simple (typically unweighted or propensity-weighted) comparison-group average with a synthetic, pre-treatment-balancing combination of donors, together with an explicit treatment of how that combination should evolve as the risk set shrinks.

\subsection{Synthetic control under staggered adoption}
\citet{ben2021synthetic} extend synthetic control methods to settings with multiple treated units adopted at different times, and introduce partial pooling of unit weights across treated units to improve the quality of the estimated counterfactual when any single treated unit's donor pool is thin. Their framework explicitly allows a donor set that varies across treated units and, implicitly, across event times, but for exposition and empirical stability they typically restrict to a fixed donor pool untreated through the relevant evaluation horizon. Our estimator can be viewed as an extension of this framework along the horizon dimension: rather than treating the donor pool as fixed once a maximum horizon is chosen, we allow it to expand at short horizons and impose an explicit continuity restriction as it contracts.

\subsection{Synthetic difference-in-differences}
\citet{arkhangelsky2021synthetic} combine unit weights and time weights with a two-way-fixed-effects-style regression to construct the SDiD estimator, which nests both synthetic control and DiD as special cases. Staggered-adoption implementations of SDiD, including sequential application of SDiD to appropriately defined cohort-specific panels, already exist in applied work. Our DiD residualization step (Section~\ref{sec:estimator}) is in the spirit of SDiD's level correction, and we adopt this proximate estimator as a natural point of comparison in our simulation design (Section~\ref{sec:simulation}), but we do not claim novelty for the residualization mechanism itself. The contribution specific to this paper is confined to how the \emph{unit weights} are constructed and propagated across a sequence of shrinking donor sets.

\subsection{Matrix completion and interactive fixed-effects estimators}
\citet{xu2017generalized} and \citet{athey2021matrix} impute untreated potential outcomes using low-rank factor structure estimated from the full panel of treated and control observations. These estimators do not construct explicit unit weights and are not restricted to convex combinations of donors, but they share with our approach the underlying interactive fixed-effects model used for identification in Section~\ref{sec:identification}. We include matrix completion as a comparator in Section~\ref{sec:simulation} because it addresses a similar problem -- imputing $Y_{it}(\infty)$ under staggered treatment -- through a different modeling route that does not explicitly confront donor-set eligibility as a weighting constraint.

\subsection{Inference for staggered synthetic control}
Cao, Lu, and Wu (2026) develop synthetic-control estimation and asymptotically valid inference for dynamic average treatment effects under staggered adoption \citep{cao2026synthetic}. Their contribution concerns identification, estimation, and asymptotically valid inference using variation from staggered treatment timing, including a test for structural instability in the fitted model. RT-SC-DiD differs by introducing explicit cross-horizon transport regularization for a sequence of shrinking admissible donor sets; the precise relationship between their donor construction and the risk-set sequence $\mathcal{D}_{g,h}$ used here requires careful comparison that we leave for a revised version of this paper, since it bears directly on the inference strategy of Section~\ref{sec:inference}. \citet{cattaneo2025scpi} develop synthetic control prediction intervals that provide valid in-sample and out-of-sample uncertainty quantification under staggered adoption, including treatment-anticipation and unit-time-level inference.

Two further papers extend SDiD specifically to dynamic, staggered settings and belong in the same comparison as RT-SC-DiD. \citet{arkhangelsky2024sequential} propose Sequential SDiD, which applies the SDiD estimator of \citet{arkhangelsky2021synthetic} sequentially to appropriately aggregated cohort data -- using estimates for early-adopting cohorts to construct counterfactuals for later ones -- and prove asymptotic equivalence to an infeasible oracle OLS estimator under a linear interactive-fixed-effects model. \citet{ciccia2024eventstudy} shows that the \citet{arkhangelsky2021synthetic} estimator can be disaggregated directly into dynamic, per-event-time treatment-effect estimators in both simple and staggered designs. Both are close competitors to RT-SC-DiD and both belong in Table~\ref{tab:comparison}.

We separately note that \citet{rho2026timeaware} propose Time-Aware Synthetic Control (TASC), which embeds synthetic control in a state-space model with a Kalman filter and Rauch--Tung--Striebel smoother to exploit temporal ordering and trend structure among pre-treatment periods. TASC addresses a different problem than the one studied here -- it targets the assumption that pre-treatment time indices are exchangeable, not horizon-dependent donor eligibility -- so we do not include it in Table~\ref{tab:comparison}.

\subsection{Comparison summary}
Table~\ref{tab:comparison} summarizes how RT-SC-DiD relates to the closest existing methods along the dimensions most relevant to the problem addressed here. The distinguishing feature of RT-SC-DiD is the combination of the last two columns: it uses the full horizon-varying set of currently eligible controls (like the not-yet-treated comparison group in \citealp{callaway2021difference} and the staggered donor pools in \citealp{ben2021synthetic}), but couples adjacent horizons through an explicit, penalized transport map with a proved distortion bound (Proposition~\ref{prop:reopt}), which none of the comparators formalize.

\begin{table}[htbp]
\centering
\small
\begin{tabularx}{\textwidth}{lXXXX}
\toprule
Method & Time-varying valid controls & Dynamic effects & Unit weights & Cross-horizon coupling \\
\midrule
Callaway--Sant'Anna \citeyearpar{callaway2021difference} & yes & yes & no (simple avg.) & no \\
Partially pooled SCM \citeyearpar{ben2021synthetic} & limited (fixed pool typical) & yes & yes & no \\
Synthetic DiD \citeyearpar{arkhangelsky2021synthetic} & no (single treated period) & via staggered variants & yes & no \\
Sequential SDiD \citeyearpar{arkhangelsky2024sequential} & yes (sequential cohort imputation) & yes & yes & sequential, oracle-equivalent \\
Event-study SDiD \citeyearpar{ciccia2024eventstudy} & yes & yes (disaggregated) & yes & no \\
Cao--Lu--Wu SCM inference \citeyearpar{cao2026synthetic} & yes & yes & yes & not by design \\
SCPI \citeyearpar{cattaneo2025scpi} & yes & yes & yes & not by design \\
RT-SC-DiD (this paper) & yes & yes & yes & transport penalty, proved bound \\
\bottomrule
\end{tabularx}
\caption{Positioning relative to closely related methods. ``Cross-horizon coupling'' refers to an explicit mechanism linking weight estimates at adjacent event-time horizons, as opposed to independent re-estimation or a fixed donor pool. Sequential SDiD's coupling is sequential imputation across cohorts (not horizons within a cohort) with a proved oracle-equivalence result; it addresses a related but distinct coupling problem from the one studied here.}
\label{tab:comparison}
\end{table}

\subsection{What we believe is new}
Existing staggered synthetic-control methods recognize that donor eligibility is horizon-dependent, and existing practice resolves this dependence either by fixing the donor pool at the longest horizon or by re-estimating weights independently at each horizon. To the best of our knowledge, no existing method develops a transport-regularized sequence of horizon-specific synthetic weights explicitly designed to isolate treatment-effect dynamics from artifacts of a contracting donor set, together with a proved bound on the resulting distortion and diagnostics quantifying it in practice. We regard this as a methodological refinement within an active and already crowded literature, not as the introduction of a fundamentally new causal-inference problem, and we return to this scoping in Section~\ref{sec:discussion}.

\section{Setup and estimands}
\label{sec:setup}

Let units be indexed by $i = 1, \dots, N$ and periods by $t = 1, \dots, T$. Let $G_i \in \{1, \dots, T, \infty\}$ denote the first period in which unit $i$ receives treatment, with $G_i = \infty$ denoting a never-treated unit. Treatment status is $D_{it} = \mathbf{1}\{t \geq G_i\}$. Let $Y_{it}(g)$ denote the potential outcome that would be observed at $t$ if treatment began at date $g$, for $g \in \{1,\dots,T,\infty\}$, and write $Y_{it}(\infty)$ for the potential outcome under no treatment through the end of the observation window; we use $Y_{it}(\infty)$ rather than the shorthand $Y_{it}(0)$ used in some of the literature throughout, to avoid any confusion between ``potential outcome under the never-treated regime'' and ``value zero,'' which matters once we later write weighted sums and differences of these potential outcomes. We maintain no anticipation, $Y_{it}(g) = Y_{it}(\infty)$ for $t < g$, so that the observed outcome is
\begin{equation}
Y_{it} = \mathbf{1}\{t < G_i\}\, Y_{it}(\infty) + \mathbf{1}\{t \geq G_i\}\, Y_{it}(G_i).
\end{equation}
Let $\mathcal{I}_g = \{i : G_i = g\}$ denote treatment cohort $g$, with cardinality $N_g$, and let $\bar Y_{g,t} = N_g^{-1} \sum_{i \in \mathcal{I}_g} Y_{it}$ denote the cohort mean at $t$. The principal estimand is the cohort-time average treatment effect on the treated,
\begin{equation}
\tau_{g,h} \equiv ATT(g, g+h) = \mathbb{E}\big[Y_{i,g+h}(g) - Y_{i,g+h}(\infty) \mid G_i = g\big], \qquad h \geq 0,
\end{equation}
following \citet{callaway2021difference}. A cohort-level summary can be formed as $ATT(g) = \sum_{t \geq g} \omega_{g,t}\, ATT(g,t)$ for weights $\omega_{g,t}$ summing to one, and overall or event-time summaries are discussed in Section~\ref{sec:aggregation}.

\paragraph{Unit-level versus cohort-level donor weighting.} The estimator below assigns synthetic weights $\gamma_{ig,h}$ to individual donor \emph{units} $i$, not to cohort means; a cohort-mean-weighted version, in which the donor pool consists of cohort averages $\bar Y_{r,t}$ for $r \neq g$ rather than unit-level series $Y_{it}$, is a special case obtained by first replacing every donor's raw series with its cohort mean and treating each cohort as a single ``super-donor.'' We work at the unit level throughout because it nests the cohort-level version, allows finer control over pre-treatment fit, and makes the leave-one-donor-out diagnostics of Section~\ref{sec:diagnostics} operate at the same resolution as the estimator itself; the identification and variance arguments in Sections~\ref{sec:identification}--\ref{sec:theory} are stated for this unit-level construction and would need to be re-derived, not merely relabeled, for a cohort-aggregated version.

\section{Horizon-dependent donor risk sets}
\label{sec:riskset}

For cohort $g$ and event time $h \geq 0$, define the valid donor \emph{risk set}
\begin{equation}
\mathcal{D}_{g,h} = \{ i : G_i > g + h \}.
\end{equation}
Never-treated units belong to every risk set. A unit with $G_i = g + r$ is a valid donor for horizons $h < r$ and becomes ineligible for $h \geq r$. The risk sets are nested by construction,
\begin{equation}
\mathcal{D}_{g,h+1} \subseteq \mathcal{D}_{g,h},
\end{equation}
and we define the \emph{attrition set}
\begin{equation}
\mathcal{A}_{g,h} = \mathcal{D}_{g,h-1} \setminus \mathcal{D}_{g,h}
\end{equation}
as the donors eligible at $h-1$ but not at $h$. This nesting property is the source of the central difficulty addressed by the paper: any estimator that uses the full risk set at each horizon must confront the fact that the composition of that risk set changes deterministically with $h$, so that naively comparing $\widehat\tau_{g,h-1}$ and $\widehat\tau_{g,h}$ conflates treatment dynamics with donor-composition change.

\subsection{Baseline: independent horizon-specific synthetic control}
Let $\gamma_{g,h} = (\gamma_{ig,h})_{i \in \mathcal{D}_{g,h}}$ denote donor weights on the simplex $\Delta(\mathcal{D}_{g,h}) = \{\gamma \geq 0 : \sum_i \gamma_i = 1\}$. Let $\mathcal{P}_g = \{g - L_g, \dots, g-1\}$ denote a pre-treatment fitting window of length $L_g$, and let $v_{g,s} \geq 0$ denote pre-treatment period weights. A horizon-by-horizon synthetic control solves
\begin{equation}
\widehat\gamma_{g,h}^{\,\mathrm{ind}} = \argmin_{\gamma \in \Delta(\mathcal{D}_{g,h})} \left\{ \sum_{s \in \mathcal{P}_g} v_{g,s} \Big(\bar Y_{g,s} - \textstyle\sum_{i \in \mathcal{D}_{g,h}} \gamma_i Y_{is}\Big)^2 + \lambda \|\gamma\|_2^2 \right\},
\end{equation}
with associated effect estimate $\widehat\tau_{g,h}^{\,\mathrm{ind}} = \bar Y_{g,g+h} - \sum_{i \in \mathcal{D}_{g,h}} \widehat\gamma_{ig,h}^{\,\mathrm{ind}} Y_{i,g+h}$. Even though every donor used at every horizon is valid, the sequence $\widehat\tau_{g,0}^{\,\mathrm{ind}}, \widehat\tau_{g,1}^{\,\mathrm{ind}}, \dots$ can move sharply whenever $\mathcal{A}_{g,h}$ removes a donor that previously carried substantial weight, purely as a mechanical consequence of re-optimization on a different simplex.

\section{The RT-SC-DiD estimator}
\label{sec:estimator}

\subsection{Transport operator}
Suppose weights $\widehat\gamma_{g,h-1}$ have already been estimated on $\mathcal{D}_{g,h-1}$. When donors in $\mathcal{A}_{g,h}$ exit, $\widehat\gamma_{g,h-1}$ is no longer a valid element of $\Delta(\mathcal{D}_{g,h})$, and it must be mapped onto the new simplex before it can serve as a stabilizing reference. The simplest map is proportional renormalization,
\begin{equation}
\widetilde\gamma_{ig,h-1} = \frac{\widehat\gamma_{ig,h-1}}{\sum_{j \in \mathcal{D}_{g,h}} \widehat\gamma_{jg,h-1}}, \qquad i \in \mathcal{D}_{g,h},
\label{eq:renorm}
\end{equation}
which redistributes an exiting donor's mass across all remaining donors in proportion to their \emph{existing} weight, irrespective of pre-treatment similarity. We instead redistribute exited mass toward donors with similar pre-treatment trajectories, while nesting \eqref{eq:renorm} as an explicit limiting case rather than an approximate one. Define the pre-treatment dissimilarity
\begin{equation}
d_{ij,g} = \sum_{s \in \mathcal{P}_g} v_{g,s} (Y_{is} - Y_{js})^2,
\end{equation}
and, for each exiting donor $j \in \mathcal{A}_{g,h}$, a redistribution kernel over the remaining risk set that combines a similarity term with the surviving donor's own pre-attrition weight,
\begin{equation}
\pi_{ij,g,h} = \frac{\widehat\gamma_{ig,h-1} \exp(-d_{ij,g}/\kappa)}{\sum_{\ell \in \mathcal{D}_{g,h}} \widehat\gamma_{\ell g,h-1} \exp(-d_{\ell j,g}/\kappa)}, \qquad i \in \mathcal{D}_{g,h},
\label{eq:kernel}
\end{equation}
with bandwidth $\kappa > 0$. The \emph{transported reference weight} is
\begin{equation}
\gamma^{\mathrm{tr}}_{ig,h} = \widehat\gamma_{ig,h-1} + \sum_{j \in \mathcal{A}_{g,h}} \widehat\gamma_{jg,h-1}\, \pi_{ij,g,h}, \qquad i \in \mathcal{D}_{g,h},
\label{eq:transport}
\end{equation}
which by construction satisfies $\sum_{i \in \mathcal{D}_{g,h}} \gamma^{\mathrm{tr}}_{ig,h} = 1$, since $\sum_i \pi_{ij,g,h} = 1$ for each $j$ and the exiting donors' total mass $\sum_{j \in \mathcal{A}_{g,h}} \widehat\gamma_{jg,h-1}$ is fully reassigned. Because the kernel \eqref{eq:kernel} weights each surviving donor $i$'s share of $j$'s exited mass by $i$'s own pre-attrition weight $\widehat\gamma_{ig,h-1}$ in addition to its similarity to $j$, letting $\kappa \to \infty$ removes the similarity term entirely and leaves
\begin{equation}
\pi_{ij,g,h} \;\longrightarrow\; \frac{\widehat\gamma_{ig,h-1}}{\sum_{\ell \in \mathcal{D}_{g,h}} \widehat\gamma_{\ell g,h-1}}, \qquad \text{independent of } j,
\end{equation}
so that $\gamma^{\mathrm{tr}}_{ig,h} \to \widehat\gamma_{ig,h-1} \big(1 + m_{g,h} / \sum_{\ell \in \mathcal{D}_{g,h}} \widehat\gamma_{\ell g,h-1}\big) = \widehat\gamma_{ig,h-1} / \sum_{\ell \in \mathcal{D}_{g,h}} \widehat\gamma_{\ell g,h-1} = \widetilde\gamma_{ig,h-1}$ exactly, where $m_{g,h} = \sum_{j \in \mathcal{A}_{g,h}} \widehat\gamma_{jg,h-1}$ is the total exiting mass. (A softmax kernel that omits the leading $\widehat\gamma_{ig,h-1}$ factor converges instead to the \emph{uniform}-across-survivors allocation as $\kappa \to \infty$, a less natural limit than proportional renormalization; \eqref{eq:kernel} is constructed specifically so that $\kappa \to \infty$ recovers \eqref{eq:renorm} exactly.) As $\kappa \to 0$, each exiting donor's mass is instead assigned entirely to its single nearest surviving donor (subject to ties), so $\kappa$ indexes a genuine continuum between global proportional preservation of the pre-attrition representation and local replacement by the most similar available donor.

\begin{remark}[Scale of the similarity metric]
The dissimilarity $d_{ij,g}$ as defined compares raw outcome levels and is therefore sensitive to differences in scale, volatility, and trend across donors that need not reflect differences in the latent structure relevant for synthetic control balance (Assumption~\ref{ass:stability}). In practice we recommend one of: (i) a level-and-change standardized version, $d_{ij,g} = \sum_{s \in \mathcal{P}_g} v_{g,s} \sigma_s^{-2}(Y_{is}-Y_{js})^2 + \eta \sum_{s} v_{g,s}\sigma_{\Delta,s}^{-2}(\Delta Y_{is} - \Delta Y_{js})^2$, with $\sigma_s^2$ and $\sigma_{\Delta,s}^2$ cross-sectional variances of levels and first differences at $s$; (ii) a residualized-trajectory distance computed after removing unit and period means, $\check Y_{is} = Y_{is} - \bar Y_{i,\mathrm{pre}} - \bar Y_{\cdot s} + \bar Y_{\cdot,\mathrm{pre}}$; or (iii) when factor loadings are separately estimated, a Mahalanobis distance $d_{ij,g} = (\widehat\lambda_i - \widehat\lambda_j)^\top \widehat\Omega^{-1} (\widehat\lambda_i - \widehat\lambda_j)$ in loading space, which ties the transport metric directly to the object that Assumption~\ref{ass:stability} requires to be balanced. We use the raw-level form in the main text only for notational simplicity and do not intend it as a recommendation.
\end{remark}

\begin{remark}[The kernel transport as a special case of feature-matching transport]
Equation~\eqref{eq:transport} redistributes each exiting donor's mass individually. A more general formulation poses the transported reference directly as the solution to a feature-matching problem: given a feature map $Z_i$ (pre-treatment outcomes, covariates, or estimated loadings) and a norm $\|\cdot\|_W$,
\begin{equation}
\gamma^{\mathrm{tr}}_{g,h} = \argmin_{\gamma \in \Delta(\mathcal{D}_{g,h})} \Big\{ \Big\| \textstyle\sum_{i \in \mathcal{D}_{g,h-1}} \widehat\gamma_{ig,h-1} Z_i - \sum_{i \in \mathcal{D}_{g,h}} \gamma_i Z_i \Big\|_W^2 + \eta\, D\big(\gamma, \widetilde\gamma_{g,h-1}\big) \Big\},
\label{eq:transport-opt}
\end{equation}
for a divergence $D$ (squared $\ell_2$, KL, or negative entropy) that anchors $\gamma$ near the renormalized surviving weights $\widetilde\gamma_{g,h-1}$ of \eqref{eq:renorm}. Problem~\eqref{eq:transport-opt} asks directly for the new valid weight vector that best preserves the previous synthetic representation \emph{in feature space}, rather than preserving it donor-by-donor. The closed-form kernel transport \eqref{eq:transport}--\eqref{eq:kernel} is a computationally convenient special case obtained when $D$ is taken to be a per-exiting-donor local reassignment rather than a single global optimization, and it coincides with an optimal transport-style solution to \eqref{eq:transport-opt} only under additional conditions (e.g., $\eta \to \infty$ with $Z_i = Y_{i,\mathcal{P}_g}$ and a locally linear cost). We use \eqref{eq:transport} in the estimator below for its closed form and low computational cost, and note \eqref{eq:transport-opt} as the more general and, we think, theoretically cleaner formulation for future work.
\end{remark}

\begin{remark}[The transport kernel is undefined at zero surviving mass, and needs a fallback]
\label{rem:zero-mass}
The kernel \eqref{eq:kernel} divides by $\sum_{\ell \in \mathcal{D}_{g,h}} \widehat\gamma_{\ell g,h-1} \exp(-d_{\ell j,g}/\kappa)$, which is exactly zero whenever \emph{all} of the pre-attrition weight is concentrated on donors in the exiting set $\mathcal{A}_{g,h}$, i.e., $\sum_{i \in \mathcal{D}_{g,h}} \widehat\gamma_{ig,h-1} = 0$. This is not a measure-zero corner case: sparse synthetic control solutions routinely place all weight on one or two donors, and if those specific donors happen to exit together, the surviving mass is exactly zero. Proportional renormalization \eqref{eq:renorm} is undefined in exactly the same case. We therefore adopt, throughout the paper -- in the estimator, in the theoretical discussion of Section~\ref{sec:theory}, and in the simulation of Section~\ref{sec:simulation-illustration} -- a base-measure-augmented kernel with a fallback that remains well defined at zero surviving mass:
\begin{equation}
\pi_{ij,g,h} = \frac{\big(\widehat\gamma_{ig,h-1} + \epsilon_0\, q_{ig,h}\big) \exp(-d_{ij,g}/\kappa)}{\sum_{\ell \in \mathcal{D}_{g,h}} \big(\widehat\gamma_{\ell g,h-1} + \epsilon_0\, q_{\ell g,h}\big) \exp(-d_{\ell j,g}/\kappa)}, \qquad i \in \mathcal{D}_{g,h},
\label{eq:kernel-fallback}
\end{equation}
for a small fallback weight $\epsilon_0 > 0$ and a base probability measure $q_{g,h}$ on $\mathcal{D}_{g,h}$ (we use the uniform measure, $q_{ig,h} = 1/|\mathcal{D}_{g,h}|$, in the implementation used for Section~\ref{sec:simulation-illustration}). Because $\epsilon_0 > 0$, the denominator of \eqref{eq:kernel-fallback} is bounded below by $\epsilon_0 \min_i q_{ig,h} \exp(-\max_{\ell} d_{\ell j,g}/\kappa) > 0$ for \emph{every} $\gamma \in \Delta(\mathcal{D}_{g,h-1})$, not merely away from some boundary, so $T_h$ is well defined and, in fact, smooth on the entire (compact) simplex once \eqref{eq:kernel-fallback} is used.

We flag an important distinction between the two kernels' $\kappa \to \infty$ limits, since it is easy to conflate them. Section~\ref{sec:estimator} showed that, under the \emph{plain} kernel \eqref{eq:kernel} with fixed pre-attrition weights, $\kappa \to \infty$ recovers proportional renormalization \eqref{eq:renorm} \emph{exactly} whenever surviving mass is positive. Under the \emph{fallback} kernel \eqref{eq:kernel-fallback} with $\epsilon_0 > 0$ held fixed, the same limit instead gives
\begin{equation}
\pi_{ij,g,h}^{\epsilon_0} \;\xrightarrow{\ \kappa \to \infty\ }\; \frac{\widehat\gamma_{ig,h-1} + \epsilon_0\, q_{ig,h}}{\sum_{\ell \in \mathcal{D}_{g,h}} \widehat\gamma_{\ell g,h-1} + \epsilon_0}, \qquad \text{not} \qquad \frac{\widehat\gamma_{ig,h-1}}{\sum_{\ell \in \mathcal{D}_{g,h}} \widehat\gamma_{\ell g,h-1}},
\end{equation}
an $\epsilon_0$-regularized version of proportional renormalization rather than proportional renormalization itself. The fallback kernel therefore does \emph{not} exactly nest \eqref{eq:renorm} at $\kappa \to \infty$ for any fixed $\epsilon_0 > 0$; it recovers \eqref{eq:renorm} only in the iterated limit $\kappa \to \infty$ followed by $\epsilon_0 \to 0$, and only when surviving mass is positive (the case in which \eqref{eq:renorm} is itself well defined). This is the price of the fallback: it buys a kernel that is well defined everywhere on the simplex (Remark~\ref{rem:zero-mass} above), at the cost of the $\kappa \to \infty$ limit no longer being an exact, single-limit statement of proportional renormalization for the $\epsilon_0$ actually used in estimation. Because the fallback removes the zero-denominator boundary entirely, it also removes the need to restrict the domain on which we ask for Lipschitz continuity of $T_h$ below (Assumption~\ref{ass:lipschitz}): with \eqref{eq:kernel-fallback}, $T_h$ is defined and, we expect, continuous on all of $\Delta(\mathcal{D}_{g,h-1})$.
\end{remark}

\subsection{Regularized horizon-specific weights}
Let $B_{g,h}(\gamma) = \sum_{s \in \mathcal{P}_g} v_{g,s} \big(\bar Y_{g,s} - \sum_{i \in \mathcal{D}_{g,h}} \gamma_i Y_{is}\big)^2$ denote pre-treatment imbalance on risk set $\mathcal{D}_{g,h}$. At $h = 0$, weights are estimated on the full initial risk set,
\begin{equation}
\widehat\gamma_{g,0} = \argmin_{\gamma \in \Delta(\mathcal{D}_{g,0})} \big\{ B_{g,0}(\gamma) + \lambda \|\gamma\|_2^2 \big\}.
\end{equation}
For $h \geq 1$, weights solve
\begin{equation}
\widehat\gamma_{g,h} = \argmin_{\gamma \in \Delta(\mathcal{D}_{g,h})} \Big\{ \underbrace{B_{g,h}(\gamma)}_{\text{pre-treatment balance}} + \underbrace{\lambda \|\gamma\|_2^2}_{\text{weight dispersion}} + \underbrace{\rho_h \|\gamma - \gamma^{\mathrm{tr}}_{g,h}\|_2^2}_{\text{cross-horizon continuity}} \Big\},
\label{eq:rtsdid-obj}
\end{equation}
where $\gamma^{\mathrm{tr}}_{g,h}$ is defined in \eqref{eq:transport}. Problem \eqref{eq:rtsdid-obj} is a convex quadratic program with linear equality and inequality constraints and is straightforward to solve numerically for each $(g,h)$. The transport penalty does not fix the weights; it allows them to change as needed to restore pre-treatment balance after attrition, while penalizing reconfiguration beyond what balance requires. The horizon-specific tuning parameter $\rho_h$ governs the balance-versus-continuity trade-off and may be selected as described in Section~\ref{sec:tuning}.

\subsection{Difference-in-differences residualization}
To remove persistent level differences between the treated cohort and the donor pool, define a weighted pre-treatment baseline $\bar Y_g^{\mathrm{pre}} = \sum_{s \in \mathcal{P}_g} a_{g,s} \bar Y_{g,s}$ and, for each donor $i$, $Y_{i,g}^{\mathrm{pre}} = \sum_{s \in \mathcal{P}_g} a_{g,s} Y_{is}$, with $\sum_s a_{g,s} = 1$. The RT-SC-DiD estimator is
\begin{equation}
\widehat\tau_{g,h}^{\,RT} = \Big(\bar Y_{g,g+h} - \bar Y_g^{\mathrm{pre}}\Big) - \sum_{i \in \mathcal{D}_{g,h}} \widehat\gamma_{ig,h} \Big(Y_{i,g+h} - Y_{i,g}^{\mathrm{pre}}\Big),
\label{eq:rtsdid}
\end{equation}
equivalently $\widehat\tau_{g,h}^{\,RT} = \bar Y_{g,g+h} - \sum_i \widehat\gamma_{ig,h} Y_{i,g+h} - \widehat\alpha_{g,h}$, with intercept $\widehat\alpha_{g,h} = \bar Y_g^{\mathrm{pre}} - \sum_i \widehat\gamma_{ig,h} Y_{i,g}^{\mathrm{pre}}$. This residualization allows the estimator to perform well when the weighted donor path reproduces the \emph{changes} in the treated cohort's outcome without exactly matching its level, in the spirit of the level correction in synthetic difference-in-differences \citep{arkhangelsky2021synthetic}.

\subsection{Partial pooling across cohorts}
\label{sec:pooling}
Estimating $\widehat\gamma_{g,h}$ separately for every cohort can be unstable when a given cohort's risk set is small. Let $q_{g,h}^{\mathrm{sep}}(\gamma_{g,h}) = [B_{g,h}(\gamma_{g,h})]^{1/2}$ denote cohort-specific imbalance, and define pooled imbalance over the set of cohorts $\mathcal{G}_h$ observed at horizon $h$,
\begin{equation}
q_h^{\mathrm{pool}}(\Gamma_h) = \left[ \sum_{\ell=1}^{L} u_\ell \Big( \sum_{g \in \mathcal{G}_h} p_{g,h} \big[\bar Y_{g,g-\ell} - \textstyle\sum_i \gamma_{ig,h} Y_{i,g-\ell}\big] \Big)^2 \right]^{1/2},
\end{equation}
where $u_\ell \geq 0$ are lag weights and $p_{g,h}$ are cohort pooling weights (e.g., proportional to $N_g$). Because $\mathcal{D}_{g,h} \neq \mathcal{D}_{g',h}$ for $g \neq g'$ in general, the cohort-specific weight vectors do not natively share a common index set, so the Frobenius norm below is not literally defined on a ``ragged'' collection of vectors of different lengths. We make this precise by zero-padding every cohort's weight vector to the full $N$-dimensional unit index space,
\begin{equation}
\gamma_{ig,h} = 0 \quad \text{for } i \notin \mathcal{D}_{g,h} \ \text{(in particular for } i \in \mathcal{I}_g\text{, so a cohort cannot serve as its own donor)},
\end{equation}
so that $\Gamma_h = (\gamma_{g,h})_{g \in \mathcal{G}_h} \in \mathbb{R}^{|\mathcal{G}_h| \times N}$ is a well-defined matrix on which $\|\cdot\|_F$ is standard, with the simplex constraints applied only to each row's eligible (nonzero) entries. Joint weights solve
\begin{equation}
\widehat\Gamma_h = \argmin_{\Gamma_h} \Big\{ \nu\, [q_h^{\mathrm{pool}}(\Gamma_h)]^2 + (1-\nu) \sum_{g \in \mathcal{G}_h} p_{g,h}\, [q_{g,h}^{\mathrm{sep}}(\gamma_{g,h})]^2 + \lambda \|\Gamma_h\|_F^2 + \rho_h \|\Gamma_h - \Gamma_h^{\mathrm{tr}}\|_F^2 \Big\},
\label{eq:pooled}
\end{equation}
where $\nu = 0$ recovers fully separate cohort estimation, $\nu = 1$ emphasizes balance for the (pooling-weighted) average treated cohort, and intermediate $\nu$ interpolates. The pooling structure itself follows \citet{ben2021synthetic}; the contribution here is its combination with horizon-specific risk sets and the transport penalty $\rho_h \|\Gamma_h - \Gamma_h^{\mathrm{tr}}\|_F^2$, so that pooling and continuity are enforced jointly rather than sequentially.

\subsection{Aggregation}
\label{sec:aggregation}
For horizon $h$, define the event-time average effect $\widehat{ATT}_h = \sum_{g \in \mathcal{G}_h} p_{g,h}\, \widehat\tau_{g,h}^{\,RT}$, with $p_{g,h} = N_g / \sum_{r \in \mathcal{G}_h} N_r$ or another chosen weighting. Because $\mathcal{G}_h$ can vary with $h$, we distinguish a \emph{varying-composition} event study, which uses all cohorts observable at each horizon, from a \emph{balanced-cohort} event study restricted to cohorts observed over the full horizon range $0, \dots, H$. An overall effect is $\widehat{ATT}^{\,\mathrm{overall}} = \sum_{h=0}^H \omega_h \widehat{ATT}_h$ for policy weights $\omega_h$ summing to one (e.g., equal event-time weights or cohort-exposure weights).

\section{Identification}
\label{sec:identification}

\begin{assumption}[No anticipation]
$Y_{it}(g) = Y_{it}(\infty)$ for all $t < g$. In particular, later-treated donors are unaffected by their own eventual treatment prior to its onset.
\end{assumption}

\begin{assumption}[No interference]
$Y_{it}(\mathbf{D}) = Y_{it}(D_{it})$; one unit's treatment status does not affect another unit's potential outcomes. This assumption is worth stating explicitly here because a later-treated donor may be indirectly exposed to policies adopted by earlier cohorts (e.g., through market-level spillovers), which would violate it.
\end{assumption}

\begin{assumption}[Synthetic-span condition]
\label{ass:span}
For each cohort $g$ and horizon $h$, there exist weights $\gamma^*_{g,h} \in \Delta(\mathcal{D}_{g,h})$ such that
\begin{equation}
\bar Y_{g,t}(\infty) = \sum_{i \in \mathcal{D}_{g,h}} \gamma^*_{ig,h} Y_{it}(\infty) + r_{g,t,h}
\end{equation}
with approximation error $r_{g,t,h}$ small relative to the treatment effect of interest, for $t$ in a neighborhood of $g + h$ and over $\mathcal{P}_g$.
\end{assumption}

\begin{assumption}[Stability of the synthetic representation across pre- and post-treatment periods]
\label{ass:stability}
Under the interactive fixed-effects model $Y_{it}(\infty) = \alpha_i + \delta_t + \lambda_i^\top f_t + \varepsilon_{it}$, the weighted donor loadings implied by $\gamma^*_{g,h}$ approximate the treated cohort's average loading, $\bar\lambda_g \approx \sum_{i \in \mathcal{D}_{g,h}} \gamma^*_{ig,h} \lambda_i$, and this approximation remains adequate at $t = g+h$, not merely on $\mathcal{P}_g$.
\end{assumption}

\begin{assumption}[Noninformative donor attrition given latent structure]
\label{ass:attrition}
Conditional on treatment timing $G_i$, pre-treatment history $\mathcal{H}_{i,g-1}$, and latent loading $\lambda_i$, future untreated innovations are mean independent of donor status:
\begin{equation}
\mathbb{E}\big[\varepsilon_{i,g+h} \mid G_i, \mathcal{H}_{i,g-1}, \lambda_i \big] = 0.
\end{equation}
This assumption clarifies that a later treatment date is not equivalent to random censoring of the donor pool: donor treatment timing may be informative about $\lambda_i$ or $\mathcal{H}_{i,g-1}$, but conditional on these, it must not predict future idiosyncratic shocks. This is analogous to conditional parallel trends in the DiD literature and should be assessed through balance checks, timing models, or sensitivity analysis rather than assumed by default.
\end{assumption}

\begin{assumption}[Synthetic support]
\label{ass:overlap}
$\inf_{\gamma \in \Delta(\mathcal{D}_{g,h})} B_{g,h}(\gamma) \leq c_{g,h}$ for an acceptably small $c_{g,h}$. When the donor pool cannot approximate the treated cohort's pre-treatment path to this tolerance, the corresponding $(g,h)$ estimate should be reported as weakly identified rather than extrapolated. We label this \emph{synthetic support} rather than ``overlap'' in the propensity-score sense, since it is an observable pre-treatment fit criterion on realized outcomes rather than a statement about the support of a treatment-assignment probability.
\end{assumption}

\paragraph{Identification, estimation, and regularization approximation, kept distinct.} These are three different statements, and conflating them overstates what is established here.

\emph{Identification (an oracle statement).} Assumptions~1--3 and~\ref{ass:span} together imply the population identity
\begin{equation}
\tau_{g,h} = \mathbb{E}\Big[ Y_{g,g+h} - \sum_{i \in \mathcal{D}_{g,h}} \gamma^*_{ig,h} Y_{i,g+h} \,\Big|\, G_i = g \Big] - \alpha^*_{g,h} - \mathbb{E}[r_{g,g+h,h} \mid G_i=g],
\end{equation}
where $\gamma^*_{g,h}$ are \emph{population} synthetic weights satisfying Assumption~\ref{ass:span} and $\alpha^*_{g,h}$ is the corresponding population intercept. This is an identification result only up to the approximation error $r_{g,g+h,h}$; it says nothing yet about any estimator.

\emph{Estimation (a consistency statement).} A separate argument is required to establish that the \emph{estimated} weights and intercept converge appropriately, e.g.\ $\widehat\gamma_{g,h} \xrightarrow{p} \gamma^*_{g,h}$ or, more weakly, $\sum_i \widehat\gamma_{ig,h} Y_{i,g+h}(\infty) - \sum_i \gamma^*_{ig,h} Y_{i,g+h}(\infty) \xrightarrow{p} 0$, together with $\widehat\alpha_{g,h} \xrightarrow{p} \alpha^*_{g,h}$. This requires standard conditions on $N$, $L_g$, the conditioning of $B_{g,h}$, and $\lambda \to 0$ at a suitable rate, none of which we prove formally in this draft; we view this as an open item rather than a solved one, distinct from identification.

\emph{Regularization approximation (a statement about $\rho_h$).} For any \emph{fixed} $\rho_h > 0$, problem~\eqref{eq:rtsdid-obj} does not target $\gamma^*_{g,h}$ exactly but a regularized pseudo-target $\gamma^\dagger_{g,h}(\rho_h)$ that trades pre-treatment balance against closeness to $\gamma^{\mathrm{tr}}_{g,h}$; $\gamma^\dagger_{g,h}(\rho_h) = \gamma^*_{g,h}$ only if $\gamma^{\mathrm{tr}}_{g,h}$ itself equals $\gamma^*_{g,h}$ or if $\rho_h \to 0$. The transport penalty is therefore a \emph{prediction regularizer} whose finite-sample bias must be weighed against its variance reduction (as in Section~\ref{sec:theory}), not an identifying restriction: it does not help identify $\tau_{g,h}$, and a large $\rho_h$ chosen for stability alone can introduce bias even at the population level if $\gamma^{\mathrm{tr}}_{g,h} \neq \gamma^*_{g,h}$. Consistency of the overall procedure therefore additionally requires $\rho_h \to 0$ (or $\rho_h$ growing slower than the rate at which $\gamma^{\mathrm{tr}}_{g,h}$ approaches $\gamma^*_{g,h}$) as the pre-treatment sample grows, alongside $\lambda \to 0$.

Under Assumptions 1--6, together with the estimation-consistency and regularization-rate conditions just described, $\widehat\tau_{g,h}^{\,RT}$ targets $\tau_{g,h}$ up to the approximation error $r_{g,t,h}$ and the sampling and regularization terms characterized in Section~\ref{sec:theory}. Consistent estimation of $\gamma^*_{g,h}$ and $\alpha_{g,h}$ is not assumed here; it is stated as an open condition, distinct from identification, that a full treatment of this estimator would need to establish.

\section{Error decomposition and theoretical propositions}
\label{sec:theory}

\subsection{Error decomposition}
Under the interactive fixed-effects model of Assumption~\ref{ass:stability}, the estimation error can be written as
\begin{equation}
\widehat\tau_{g,h}^{\,RT} - \tau_{g,h} = \underbrace{\Big[\bar Y_{g,g+h}(\infty) - \textstyle\sum_i \widehat\gamma_{ig,h} Y_{i,g+h}(\infty) - \widehat\alpha_{g,h}\Big]}_{\text{counterfactual error}},
\end{equation}
which decomposes further into a systematic factor-imbalance term and stochastic terms,
\begin{equation}
\Big(\bar\lambda_g - \textstyle\sum_i \widehat\gamma_{ig,h} \lambda_i\Big)^\top \big(f_{g+h} - \bar f_g^{\mathrm{pre}}\big) \;+\; \bar\varepsilon_{g,g+h} - \textstyle\sum_i \widehat\gamma_{ig,h} \varepsilon_{i,g+h} - \bar\varepsilon_g^{\mathrm{pre}} + \textstyle\sum_i \widehat\gamma_{ig,h} \varepsilon_{i,g}^{\mathrm{pre}}.
\end{equation}
The first term is bias from imperfect factor-loading balance; the remainder is sampling noise averaged over the treated cohort and the (weighted) donor pool at both the evaluation date and the pre-treatment baseline. The following schematic inequality summarizes the components that a formal finite-sample bound would need to control; we have not derived the constants $C_{1,h},\dots,C_{4,h}$ from the model above, so we present it as an organizing decomposition rather than a proved theorem, and do not count it among the results established in this paper:
\begin{equation}
\big|\widehat\tau_{g,h}^{RT} - \tau_{g,h}\big| \;\overset{\text{schematic}}{\leq}\; C_{1,h}\, q_{g,h} + C_{2,h}\, a_{g,h} + C_{3,h}\, r_{g,h} + C_{4,h}\, n_{g,h},
\end{equation}
where $q_{g,h} = \| \bar{\mathbf{Y}}_{g,\mathrm{pre}} - \mathbf{Y}_{\mathcal{D}_{g,h},\mathrm{pre}}^\top \widehat\gamma_{g,h} \|$ is pre-treatment imbalance, $a_{g,h} = \|\widehat\gamma_{g,h} - \gamma_{g,h}^{\mathrm{tr}}\|$ is donor-attrition instability, $r_{g,h}$ is synthetic-span approximation error (Assumption~\ref{ass:span}), and $n_{g,h}$ is sampling noise. The role of the transport penalty $\rho_h$ is to trade off $q_{g,h}$ against $a_{g,h}$: larger $\rho_h$ reduces $a_{g,h}$ (and hence bounds the same-period reoptimization distortion, Proposition~\ref{prop:reopt} below) but may increase $q_{g,h}$ by preventing full re-optimization for balance; smaller $\rho_h$ does the reverse. A $\rho_h$ selected to minimize estimated out-of-sample prediction error, as in Section~\ref{sec:tuning}, targets this trade-off directly.

\subsection{Limiting and nested cases}

The following three observations record how RT-SC-DiD relates to existing practice at limiting parameter values. We state them as a single lemma rather than as separate propositions, since each is an immediate consequence of the definitions in Section~\ref{sec:estimator} rather than a result requiring proof.

\begin{lemma}[Nested special cases]
\label{lem:nesting}
\begin{enumerate}[label=(\alph*)]
\item \emph{Fixed donor pool.} If no donor becomes treated over the evaluation window, so that $\mathcal{D}_{g,h} = \mathcal{D}_g$ for all $h \leq H$, the transport penalty is never active on a changing risk set, and $\widehat\tau_{g,h}^{\,RT}$ reduces to a cohort-specific synthetic-control-with-DiD-adjustment estimator applied to a fixed donor pool throughout.
\item \emph{Independent horizon-specific fits.} If $\rho_h = 0$ for all $h$, the third term of \eqref{eq:rtsdid-obj} vanishes identically, so $\widehat\gamma_{g,h} = \widehat\gamma_{g,h}^{\,\mathrm{ind}}$ and RT-SC-DiD coincides with independently fitted horizon-specific synthetic controls.
\item \emph{Fixed-risk-set restriction.} If the donor pool used at every horizon $h \leq H$ is restricted by fiat to $\mathcal{D}_{g,H} = \{i : G_i > g+H\}$, RT-SC-DiD reduces to the conventional practice of restricting to donors untreated through the longest evaluation horizon.
\end{enumerate}
\end{lemma}
These cases are useful as implementation checks and as a way of confirming that RT-SC-DiD strictly generalizes the two conventional alternatives described in Section~\ref{sec:intro}, but they do not by themselves establish that intermediate $\rho_h$ improves on either boundary; that requires the risk comparison discussed after Theorem~\ref{thm:propagation} below.

\subsection{Same-period reoptimization distortion}

A point of emphasis: different horizons correspond to different calendar dates, $g+h-1$ versus $g+h$, and the untreated potential outcome $Y_{i,g+h}(\infty)$ can move sharply between adjacent dates for reasons that have nothing to do with donor composition. A change in $\widehat\tau_{g,h} - \widehat\tau_{g,h-1}$ is therefore \emph{not} automatically evidence of an artifact, and describing the goal of the transport penalty as producing a ``smooth'' or ``continuous'' counterfactual path invites exactly this misreading: the object worth controlling is not smoothness of $\widehat\tau_{g,\cdot}$ across $h$, but the part of the horizon-$h$ counterfactual that is attributable to \emph{reweighting at a fixed calendar date} rather than to new outcome data. We therefore define, at the single evaluation date $g+h$, the two counterfactual values implied by the reoptimized weights and by the transported reference,
\begin{equation}
\widehat Y_{g,g+h}^{\,\mathrm{new}} = \sum_{i \in \mathcal{D}_{g,h}} \widehat\gamma_{ig,h} Y_{i,g+h}, \qquad \widehat Y_{g,g+h}^{\,\mathrm{tr}} = \sum_{i \in \mathcal{D}_{g,h}} \gamma^{\mathrm{tr}}_{ig,h} Y_{i,g+h},
\end{equation}
and define the \emph{same-period reoptimization distortion}
\begin{equation}
R_{g,h} = \widehat Y_{g,g+h}^{\,\mathrm{new}} - \widehat Y_{g,g+h}^{\,\mathrm{tr}} = \sum_{i \in \mathcal{D}_{g,h}} \big(\widehat\gamma_{ig,h} - \gamma^{\mathrm{tr}}_{ig,h}\big) Y_{i,g+h}.
\end{equation}
Both terms are evaluated at the same calendar date $g+h$ using the same outcome data, so $R_{g,h}$ isolates the effect of \emph{who is reweighted and by how much}, holding the outcome realizations fixed; it is exactly the diagnostic $J_{g,h}$ of Section~\ref{sec:diagnostics} up to sign. We adopt this same-period framing, rather than a cross-horizon smoothness claim, as the object our main distortion bound (Proposition~\ref{prop:reopt} below) controls.

\begin{proposition}[Bounded reoptimization distortion]
\label{prop:reopt}
Suppose $|Y_{it}| \leq M$ for all $i,t$ (observed outcomes; the population statement under $Y_{it}(\infty)$ is identical). Then
\begin{equation}
\big| R_{g,h} \big| \;\leq\; M \, \big\| \widehat\gamma_{g,h} - \gamma^{\mathrm{tr}}_{g,h} \big\|_1 \;=\; M\, T_{g,h}^{(1)}.
\label{eq:l1bound}
\end{equation}
Writing $K_{g,h} = |\mathcal{D}_{g,h}|$ and using $\|x\|_1 \leq \sqrt{K_{g,h}}\, \|x\|_2$ gives the dimension-adjusted bound in terms of the $\ell_2$ transport distance $T_{g,h} = \|\widehat\gamma_{g,h} - \gamma^{\mathrm{tr}}_{g,h}\|_2$ used in the optimization penalty,
\begin{equation}
\big| R_{g,h} \big| \;\leq\; M \sqrt{K_{g,h}}\; T_{g,h}.
\label{eq:l2bound}
\end{equation}
\end{proposition}
\begin{proof}
By Hölder's inequality, $|\sum_i c_i Y_i| \leq \|c\|_1 \max_i |Y_i| \leq M\|c\|_1$ for $c_i = \widehat\gamma_{ig,h} - \gamma^{\mathrm{tr}}_{ig,h}$, giving \eqref{eq:l1bound} directly. Inequality \eqref{eq:l2bound} follows from the standard norm inequality $\|c\|_1 \leq \sqrt{K_{g,h}}\|c\|_2$ on $\mathbb{R}^{K_{g,h}}$, which holds with equality when $|c_i|$ is constant across coordinates and is otherwise strict.
\end{proof}
Because the optimization problem \eqref{eq:rtsdid-obj} penalizes $\|\gamma - \gamma^{\mathrm{tr}}_{g,h}\|_2^2$, the natural diagnostic to report alongside the theorem is the $\ell_1$ quantity $T_{g,h}^{(1)}$ that maps directly to \eqref{eq:l1bound}; the $\ell_2$ quantity $T_{g,h}$ used inside the objective satisfies only the weaker, dimension-dependent bound \eqref{eq:l2bound}, since $K_{g,h}$ can be large when the donor pool is large. Section~\ref{sec:diagnostics} reports both.

\begin{remark}[Exact recovery under convex-hull preservation, informal]
\label{rem:recovery}
If $\bar\lambda_g \in \conv\{\lambda_i : i \in \mathcal{D}_{g,h}\}$ at every horizon and the pre-treatment window ``identifies the factor-loading representation,'' one would expect that as $L_g \to \infty$ the untreated counterfactual $\sum_i \gamma^*_{ig,h} Y_{it}(\infty)$ recovers $\bar Y_{g,t}(\infty)$ asymptotically. We state this only as an informal remark, not a proposition, because the phrase ``identifies the factor-loading representation'' is carrying essentially the entire result and has not been made precise here: a formal version would need to specify the factor rank, the variation and normalization of $f_t$, the behavior of the idiosyncratic errors $\varepsilon_{it}$, whether the recovering weights are unique or merely predictively equivalent, how the DiD intercept correction interacts with the recovery argument, whether $N$ is held fixed or grows alongside $L_g$, and how $L_g$, $\lambda$, and $\rho_h$ must jointly behave. We leave a fully specified version of this result, with these conditions made explicit, to future work.
\end{remark}

\begin{proposition}[Bias from naive renormalization]
\label{prop:naive}
Consider a single donor $j$ exiting between $h-1$ and $h$ with pre-attrition weight $w_j = \widehat\gamma_{jg,h-1} > 0$. Let $\mu_{h-1} = \sum_{i \in \mathcal{D}_{g,h-1}} \widehat\gamma_{ig,h-1} \lambda_i$ denote the pre-exit synthetic loading, and let $\mu_{-j} = (1-w_j)^{-1} \sum_{i \neq j} \widehat\gamma_{ig,h-1} \lambda_i$ denote the weighted average loading of the surviving donors under their \emph{original} (pre-renormalization) weights. Since $\mu_{h-1} = w_j \lambda_j + (1-w_j)\mu_{-j}$, proportional renormalization \eqref{eq:renorm} yields synthetic loading $\mu_{-j}$, and
\begin{equation}
\mu_{-j} - \mu_{h-1} = w_j\,(\mu_{-j} - \lambda_j), \qquad \text{so} \qquad \|\mu_{-j} - \mu_{h-1}\| = w_j\, \|\mu_{-j} - \lambda_j\|.
\end{equation}
Consequently, proportional renormalization preserves the pre-exit synthetic loading exactly if and only if $w_j = 0$ or $\lambda_j = \mu_{-j}$; whenever the exiting donor's loading differs from the surviving-weighted average, renormalization introduces loading bias of exact magnitude $w_j \|\mu_{-j} - \lambda_j\|$, even though the pre-exit representation was assumed exactly balanced. More generally, for a set of simultaneously exiting donors $\mathcal{A}_{g,h}$ with total mass $m = \sum_{j \in \mathcal{A}_{g,h}} \widehat\gamma_{jg,h-1}$, exiting-donor average loading $\mu_A = m^{-1}\sum_{j \in \mathcal{A}_{g,h}} \widehat\gamma_{jg,h-1}\lambda_j$, and surviving-donor average loading $\mu_S = (1-m)^{-1}\sum_{i \in \mathcal{D}_{g,h}} \widehat\gamma_{ig,h-1}\lambda_i$, the identity $\mu_{h-1} = m\mu_A + (1-m)\mu_S$ gives
\begin{equation}
\mu_S - \mu_{h-1} = m\,(\mu_S - \mu_A),
\end{equation}
so the loading bias from renormalization scales with the product of exited weight mass $m$ and the dissimilarity between exiting- and surviving-donor average loadings, $\|\mu_S - \mu_A\|$.
\end{proposition}
This directly motivates the two diagnostics of Section~\ref{sec:diagnostics}, $M_{g,h}^{\mathrm{exit}} = m$ and the companion trajectory-discrepancy statistic $\Lambda_{g,h}$ tracking $\|\mu_S - \mu_A\|$ (defined precisely in Section~\ref{sec:diagnostics} via weighted mean pre-treatment trajectories, since $\lambda_i$ is not directly observed). The similarity-weighted transport of \eqref{eq:transport} is designed so that mass exiting with loading $\lambda_j$ is redirected preferentially toward surviving donors with similar pre-treatment trajectories, which under Assumption~\ref{ass:stability} approximates matching on $\lambda_j$ and so is designed to reduce, though not in general eliminate, this bias relative to uniform renormalization.

\subsection{Error propagation across horizons}

The transport recursion $\widehat\gamma_{g,h-1} \to \gamma^{\mathrm{tr}}_{g,h} \to \widehat\gamma_{g,h}$ links every horizon's estimate to the previous one, so estimation error at an early horizon can propagate forward rather than being confined to the horizon at which it occurs. This distinguishes RT-SC-DiD from horizon-independent estimation ($\rho_h = 0$) and is a cost, not only a benefit, of the transport mechanism; it deserves a formal bound rather than being treated solely as a stabilizing feature.

Let $\gamma^*_{g,h}$ denote the oracle weights of Assumption~\ref{ass:span}, let $e_{g,h} = \widehat\gamma_{g,h} - \gamma^*_{g,h}$, and let $T_h(\gamma) \equiv \gamma^{\mathrm{tr}}_{g,h}(\gamma)$ denote the map from horizon-$(h-1)$ weights to the transported reference, defined throughout via \eqref{eq:transport} together with the fallback kernel \eqref{eq:kernel-fallback} of Remark~\ref{rem:zero-mass} (rather than the plain kernel \eqref{eq:kernel}, which can be undefined at zero surviving mass and so cannot serve as the domain for a Lipschitz assumption on all of $\Delta(\mathcal{D}_{g,h-1})$). We emphasize that $T_h$ is \emph{not} linear in its argument: $\pi_{ij,g,h}$ in \eqref{eq:kernel-fallback} has $\gamma$ appearing in both the numerator and the denominator, so $T_h$ is a ratio of linear functions of $\gamma$ and is therefore nonlinear in general; $T_h$ does not have a well-defined operator norm, and the recursion below requires Lipschitz continuity of $T_h$, stated as an explicit regularity condition rather than derived from linearity.

\begin{assumption}[Regularity of the transport map]
\label{ass:lipschitz}
$T_h$, defined via the fallback kernel \eqref{eq:kernel-fallback} with fallback weight $\epsilon_0 > 0$, is Lipschitz on the simplex $\Delta(\mathcal{D}_{g,h-1})$ with constant $L_h < \infty$: $\|T_h(\gamma) - T_h(\gamma')\| \leq L_h \|\gamma - \gamma'\|$ for all $\gamma, \gamma' \in \Delta(\mathcal{D}_{g,h-1})$.
\end{assumption}
Because $\Delta(\mathcal{D}_{g,h-1})$ is compact and, with $\epsilon_0 > 0$, the kernel denominator in \eqref{eq:kernel-fallback} is bounded away from zero everywhere on it (Remark~\ref{rem:zero-mass}), $T_h$ is a ratio of smooth functions with a non-vanishing denominator on a compact domain, which is the standard setting in which a finite Lipschitz constant is expected to exist; we do not, however, derive an explicit value of $L_h$ for \eqref{eq:kernel-fallback} in this paper. Theorem~\ref{thm:propagation} below is accordingly a conditional stability result given Assumption~\ref{ass:lipschitz}, not an unconditional property of \eqref{eq:kernel-fallback} that we have established from first principles.

\begin{theorem}[Recursive error propagation]
\label{thm:propagation}
Define the current-horizon optimization error $u_{g,h} = \|\widehat\gamma_{g,h} - T_h(\widehat\gamma_{g,h-1})\|$ and the transport approximation error $b^{\mathrm{tr}}_{g,h} = \|T_h(\gamma^*_{g,h-1}) - \gamma^*_{g,h}\|$ (the extent to which transporting the \emph{oracle} weights forward fails to reproduce the oracle weights at the new horizon). Then
\begin{equation}
e_{g,h} = \underbrace{\big[\widehat\gamma_{g,h} - T_h(\widehat\gamma_{g,h-1})\big]}_{\text{current error}} + \underbrace{\big[T_h(\widehat\gamma_{g,h-1}) - T_h(\gamma^*_{g,h-1})\big]}_{\text{propagated prior error}} + \underbrace{\big[T_h(\gamma^*_{g,h-1}) - \gamma^*_{g,h}\big]}_{\text{transport approximation error}},
\end{equation}
and, applying the triangle inequality and the Lipschitz property to the middle term,
\begin{equation}
\|e_{g,h}\| \;\leq\; u_{g,h} + L_h \|e_{g,h-1}\| + b^{\mathrm{tr}}_{g,h}.
\end{equation}
Unrolling the recursion from $h=0$ gives
\begin{equation}
\|e_{g,h}\| \;\leq\; \sum_{r=1}^{h} \Big(\prod_{q=r+1}^{h} L_q\Big) \big(u_{g,r} + b^{\mathrm{tr}}_{g,r}\big) \;+\; \Big(\prod_{q=1}^{h} L_q\Big) \|e_{g,0}\|.
\label{eq:propagation}
\end{equation}
\end{theorem}
\begin{proof}
The decomposition of $e_{g,h}$ follows by adding and subtracting $T_h(\widehat\gamma_{g,h-1})$ and $T_h(\gamma^*_{g,h-1})$. The one-step bound follows from the triangle inequality applied to the three bracketed terms and the Lipschitz bound on the middle term. Inequality \eqref{eq:propagation} follows by induction on $h$, substituting the one-step bound for $\|e_{g,h-1}\|$ repeatedly.
\end{proof}
If $L_q \leq 1$ for all $q$ -- which we have not derived or numerically established for the kernel \eqref{eq:kernel-fallback}, and state only as a conditional case since redistribution of probability mass by a normalized nonlinear map does not automatically bound its Lipschitz constant by one -- then \eqref{eq:propagation} implies that prior error is not multiplicatively amplified across horizons: error does not explode, though it does not vanish either, since it accumulates the per-horizon optimization and transport-approximation errors without discounting. The relevant contrast with $\rho_h = 0$ is not that estimation errors become statistically independent across horizons -- they generally remain dependent under either choice of $\rho_h$, since the horizon-specific estimators share the same treated cohort, overlapping donor units, overlapping pre-treatment outcomes, and correlated post-treatment shocks, none of which the transport mechanism creates or removes. The distinction $\rho_h = 0$ actually buys is the absence of \emph{recursive transmission through the estimator itself}: with $\rho_h = 0$, $\widehat\gamma_{g,h}$ does not depend functionally on $\widehat\gamma_{g,h-1}$, so there is no term analogous to $L_h \|e_{g,h-1}\|$ linking the two estimation errors mechanically, even though the underlying data-generating dependence across horizons remains. This is the correct basis for comparing $\rho_h > 0$ against $\rho_h = 0$: transport reduces the variance component of $u_{g,h}$ (by anchoring the horizon-$h$ optimization, which shrinks its effective search space) at the cost of introducing mechanical propagation through the $L_h\|e_{g,h-1}\|$ term and a potential bias through $b^{\mathrm{tr}}_{g,h}$ whenever the oracle weights themselves are not exactly transport-consistent, i.e., whenever $T_h(\gamma^*_{g,h-1}) \neq \gamma^*_{g,h}$. A more complete risk comparison would formalize this trade-off under a local-stability model for the oracle path, e.g., $\|\gamma^*_{g,h} - T_h(\gamma^*_{g,h-1})\| \leq \delta_{g,h}$ for small $\delta_{g,h}$, and derive conditions on $(\delta_{g,h}, \rho_h)$ under which expected squared error under transport is lower than under independent re-estimation; we state this as a natural next step rather than a proved result here, since it requires a variance model for $u_{g,h}$ as a function of $\rho_h$ that we have not yet derived. Section~\ref{sec:simulation-illustration} reports a first, purely empirical (simulation-based) look at this trade-off in place of the still-missing analytical risk comparison.

\section{Tuning-parameter selection}
\label{sec:tuning}

The tuning parameters $(\lambda, \rho_h, \nu, \kappa)$ should be chosen without reference to the treated cohort's post-treatment outcomes. We propose a donor-only placebo-validation procedure:
\begin{enumerate}[itemsep=1pt]
\item select a donor untreated through a pseudo-evaluation horizon;
\item assign it a placebo adoption date reproducing a plausible attrition pattern (e.g., matched to the empirical distribution of $G_i - g$ across real donors of the cohort under study);
\item reconstruct its subsequent untreated outcomes using the remaining donors, following the same estimator with candidate parameters $\theta = (\lambda,\rho,\nu,\kappa)$;
\item compute horizon-specific prediction error against the placebo unit's actual (untreated) outcomes;
\item repeat over many donors and placebo dates, and choose $\theta$ minimizing average placebo loss,
\begin{equation}
CV(\theta) = \sum_{b=1}^B \sum_{h=0}^{H_b} \zeta_h \Big[ Y_{i_b, g_b+h} - \widehat Y_{i_b,g_b+h}^{(-i_b)}(0;\theta) \Big]^2.
\end{equation}
\end{enumerate}
This directly targets the estimator's central task -- predicting untreated outcomes as the eligible donor pool contracts -- rather than an auxiliary or unrelated loss.

\section{Inference}
\label{sec:inference}

\paragraph{Cohort-level bootstrap.} Resample independent units or clusters, reconstruct all risk sets, transport operators, and weights, and recompute the full estimator. This captures uncertainty from cohort composition, donor selection, weight estimation, and event-time aggregation jointly. We flag that a naive i.i.d.\ unit (or cluster) bootstrap is \emph{not} automatically valid for synthetic-control-type estimators: validity typically requires either a large number of independent clusters, or an asymptotic regime in which both the treated cohort and the donor pool grow, and is known to fail in some small-$N$ synthetic control settings even without the additional complication of a horizon-varying donor pool. We do not present the bootstrap as validated for RT-SC-DiD by default; the recent staggered-synthetic-control inference literature \citep{cao2026synthetic,cattaneo2025scpi} develops conditions under which specific resampling and prediction-interval schemes are valid in exactly this class of designs, and a full treatment of RT-SC-DiD inference should build on, rather than bypass, that literature.

\paragraph{Placebo and permutation inference.} Assign placebo adoption dates to untreated or not-yet-treated units and compare the observed statistic to its placebo distribution, e.g., via the standardized statistic $S_{g,h} = \widehat\tau_{g,h} / \widehat\sigma_{g,h}^{\,\mathrm{placebo}}$.

\paragraph{Conformal counterfactual intervals.} Following \citet{chernozhukov2021exact}, pre-treatment residuals and untreated placebo prediction errors can be used to construct prediction intervals for $Y_{g,g+h}(\infty)$, translating into intervals for $\tau_{g,h}$. We avoid describing these as ``distribution-free'' without qualification: \citet{chernozhukov2021exact}'s exact-validity guarantee requires the estimated residual sequence to be exchangeable (e.g., approximately satisfied under i.i.d.\ data), and their approximate-validity guarantee under serial dependence requires additional, method-specific conditions that would need to be verified for the RT-SC-DiD transport recursion specifically, since the recursion in Theorem~\ref{thm:propagation} induces dependence across horizons that a generic exchangeability assumption does not automatically accommodate.

\paragraph{Influence-function augmentation.} A secondary extension augments $\widehat\tau_{g,h}^{\,RT}$ with a cross-fitted outcome-regression correction $\widehat R_{g,h}$ for residual imbalance, $\widehat\tau_{g,h}^{\,AUG} = \widehat\tau_{g,h}^{\,RT} + \widehat R_{g,h}$, in the spirit of augmented / doubly robust DiD estimators. We present this as a direction for future work rather than a fully developed result, pending a formal double-robustness proof under Assumptions~\ref{ass:span}--\ref{ass:attrition}.

\section{Donor-support diagnostics}
\label{sec:diagnostics}

We recommend reporting the following quantities alongside every estimated $\widehat\tau_{g,h}^{\,RT}$.

\paragraph{Effective donor count.} $N_{g,h}^{\mathrm{eff}} = \big( \sum_{i \in \mathcal{D}_{g,h}} \widehat\gamma_{ig,h}^2 \big)^{-1}$, the inverse Herfindahl index of the weight vector.

\paragraph{Attrited weight mass.} $M_{g,h}^{\mathrm{exit}} = \sum_{i \in \mathcal{A}_{g,h}} \widehat\gamma_{ig,h-1}$, the fraction of the previous horizon's synthetic control invalidated at $h$; this is the mass term $m$ in Proposition~\ref{prop:naive}.

\paragraph{Exit-survivor loading discrepancy.} Since $\lambda_i$ is not directly observed, we approximate $\|\mu_S - \mu_A\|$ from Proposition~\ref{prop:naive} using weighted mean pre-treatment \emph{trajectories} rather than the pairwise scalar distance $d_{ij,g}$, which is not itself a vector and so cannot be differenced the way a trajectory can. Let $Y_{i,\mathrm{pre}} = (Y_{is})_{s \in \mathcal{P}_g}$ denote donor $i$'s pre-treatment trajectory, and define the exiting- and surviving-donor weighted mean trajectories under the pre-attrition weights,
\begin{equation}
\bar Y_{A,g,h} = \frac{1}{m_{g,h}} \sum_{j \in \mathcal{A}_{g,h}} \widehat\gamma_{jg,h-1}\, Y_{j,\mathrm{pre}}, \qquad \bar Y_{S,g,h} = \frac{1}{1-m_{g,h}} \sum_{i \in \mathcal{D}_{g,h}} \widehat\gamma_{ig,h-1}\, Y_{i,\mathrm{pre}},
\end{equation}
with $m_{g,h} = M_{g,h}^{\mathrm{exit}}$ as above, and set $\Lambda_{g,h} = \| \bar Y_{S,g,h} - \bar Y_{A,g,h} \|_V$ for a chosen norm $V$ (e.g., the $v_{g,s}$-weighted Euclidean norm used elsewhere in the paper). Proposition~\ref{prop:naive} implies that $M_{g,h}^{\mathrm{exit}} \times \Lambda_{g,h}$, rather than $M_{g,h}^{\mathrm{exit}}$ alone, is the quantity that governs naive-renormalization bias (with $\Lambda_{g,h}$ standing in for the unobserved $\|\mu_S - \mu_A\|$), so we recommend reporting the product alongside each component.

\paragraph{Transport distance.} We report both the $\ell_1$ transport distance $T_{g,h}^{(1)} = \|\widehat\gamma_{g,h} - \gamma^{\mathrm{tr}}_{g,h}\|_1$, which by Proposition~\ref{prop:reopt} bounds the same-period reoptimization distortion $|R_{g,h}|$ directly (Eq.~\eqref{eq:l1bound}), and the $\ell_2$ transport distance $T_{g,h} = \|\widehat\gamma_{g,h} - \gamma^{\mathrm{tr}}_{g,h}\|_2$ used inside the optimization penalty in \eqref{eq:rtsdid-obj}, which bounds $|R_{g,h}|$ only after multiplying by $\sqrt{K_{g,h}}$ (Eq.~\eqref{eq:l2bound}). Reporting $T_{g,h}^{(1)}$ alongside $T_{g,h}$ avoids overstating what the $\ell_2$ quantity alone implies about counterfactual distortion.

\paragraph{Same-period reoptimization distortion.} $R_{g,h} = \sum_i (\widehat\gamma_{ig,h} - \gamma^{\mathrm{tr}}_{ig,h}) Y_{i,g+h}$, evaluated at the single calendar date $g+h$ (Section~\ref{sec:theory}); a large value warns that part of an estimated change in $\widehat\tau_{g,h}$ relative to $\widehat\tau_{g,h-1}$ may be attributable to reweighting rather than to new outcome data, as distinct from genuine movement in $Y_{i,g+h}(\infty)$ itself, which $R_{g,h}$ does not and is not intended to flag.

\paragraph{Pre-treatment RMSPE.} $RMSPE_{g,h} = \big( L_g^{-1} \sum_{s \in \mathcal{P}_g} (\bar Y_{g,s} - \sum_i \widehat\gamma_{ig,h} Y_{is})^2 \big)^{1/2}$.

We suggest that every event-study figure produced under RT-SC-DiD be accompanied by a companion support panel plotting $N_{g,h}^{\mathrm{eff}}$, $M_{g,h}^{\mathrm{exit}}$, $T_{g,h}^{(1)}$, $T_{g,h}$, and $RMSPE_{g,h}$ against $h$, so that readers can assess how much of the reported dynamics rests on a well-supported donor pool versus a thin or rapidly changing one.

\section{Monte Carlo design}
\label{sec:simulation}

\subsection{Data-generating process}
We propose an untreated potential-outcome model
\begin{equation}
Y_{it}(\infty) = \alpha_i + \delta_t + \lambda_i^\top f_t + \beta^\top X_{it} + \varepsilon_{it},
\end{equation}
with treatment-effect paths either linear, $\tau_{i,h} = \tau_0 + \tau_1 h + \eta_i$, or saturating, $\tau_{i,h} = \tau_{\max}(1 - e^{-\psi h}) + \eta_i$. Treatment timing among not-yet-treated units may depend on latent structure,
\begin{equation}
\Pr(G_i = t \mid G_i \geq t) = \mathrm{logit}^{-1}(a_t + b^\top X_i + c^\top \lambda_i),
\end{equation}
allowing us to control the extent to which donor attrition is informative (Assumption~\ref{ass:attrition}) versus effectively random.

\subsection{Design factors}
We propose varying: (i) the never-treated share ($0\%, 10\%, 30\%, 50\%$); (ii) speed of staggered rollout (slow, moderate, rapid); (iii) the pattern of donor attrition (diffuse, concentrated at a single horizon, concentrated among high-weight donors); (iv) latent-factor dimension ($F = 1, 3, 5$); (v) pre-treatment length ($L = 5, 10, 20, 40$); (vi) the basis for treatment-timing selection (random, observed-covariate-based, latent-factor-based, recent-shock-based); (vii) treatment-effect heterogeneity (homogeneous, cohort-heterogeneous, unit-heterogeneous, calendar-time-heterogeneous); and (viii) spillovers (absent, weak, moderate), the last of which is included specifically to probe sensitivity to violations of the no-interference assumption.

\subsection{Comparator estimators}
We propose comparing RT-SC-DiD against: Callaway--Sant'Anna with not-yet-treated controls, with and without covariate adjustment \citep{callaway2021difference}; standard and staggered synthetic DiD \citep{arkhangelsky2021synthetic}; partially pooled staggered synthetic control \citep{ben2021synthetic}; matrix completion \citep{xu2017generalized,athey2021matrix}; the staggered synthetic control inference approach of \citet{cao2026synthetic} and synthetic control prediction intervals \citep{cattaneo2025scpi}; fixed-final-horizon synthetic control; independently re-estimated horizon-specific synthetic control ($\rho_h = 0$); and naive deletion-and-renormalization, implemented as a separate, non-reoptimizing comparator that deletes newly ineligible donors and rescales surviving weights proportionally via \eqref{eq:renorm} rather than as a limiting case of the RT-SC-DiD objective. (As $\rho_h \to \infty$ in \eqref{eq:rtsdid-obj}, the RT-SC-DiD optimizer converges to the similarity-based transported reference $\gamma^{\mathrm{tr}}_{g,h}$ itself, which coincides with naive proportional renormalization only in the further, iterated limit $\kappa \to \infty$, $\epsilon_0 \to 0$ discussed in Remark~\ref{rem:zero-mass}; the two are not the same comparator in general, so we implement naive renormalization directly rather than via $\rho_h \to \infty$.)

\subsection{Evaluation criteria}
For each $(g,h)$ we propose reporting bias, RMSE, coverage, and interval length for $\tau_{g,h}$, together with pre-treatment RMSPE, effective donor count, and the same-period reoptimization distortion $R_{g,h}$. A metric specific to this paper's concern is a decomposed \emph{counterfactual discontinuity index}. The naive version,
\begin{equation}
CDI^{\mathrm{total}} = \sum_{g,h} \big| \Delta \widehat Y_{g,h}(\infty) - \Delta Y_{g,h}(\infty) \big|,
\end{equation}
comparing estimated and true horizon-to-horizon changes in the imputed untreated path, conflates two distinct sources of error: ordinary prediction error in $\widehat Y_{g,h}(\infty)$ that would occur even under a fixed donor pool, and error specifically attributable to donor-composition change. To isolate the latter we report $CDI^{\mathrm{total}}$ alongside two components benchmarked against the \emph{oracle-transported} counterfactual $Y_{g,g+h}^{\mathrm{tr},*} = \sum_i \gamma^{\mathrm{tr},*}_{ig,h} Y_{i,g+h}(\infty)$ (the transport of the true oracle weights, which is well defined in simulation where $\gamma^*_{g,h}$ is known):
\begin{align}
CDI^{\mathrm{fit}} &= \sum_{g,h} \big| \widehat Y^{\,\mathrm{tr}}_{g,g+h} - Y_{g,g+h}^{\mathrm{tr},*} \big|, \\
CDI^{\mathrm{reopt}} &= \sum_{g,h} \big| R_{g,h} - R^*_{g,h} \big|, \qquad R^*_{g,h} = \sum_i \big(\gamma^*_{ig,h} - \gamma^{\mathrm{tr},*}_{ig,h}\big) Y_{i,g+h}(\infty),
\end{align}
so that $CDI^{\mathrm{fit}}$ isolates ordinary estimation error in the transported reference and $CDI^{\mathrm{reopt}}$ isolates excess reoptimization distortion relative to what the oracle weights themselves would require. This decomposition is only available in simulation, where $\gamma^*_{g,h}$ is known by construction; in the empirical application of Section~\ref{sec:application} we rely instead on the observable diagnostics of Section~\ref{sec:diagnostics} as proxies.

\subsection{Illustrative pilot simulation}
\label{sec:simulation-illustration}

Sections~\ref{sec:simulation}.1--\ref{sec:simulation}.3 describe the full design we consider appropriate for a completed working paper. Here we report a single, deliberately small pilot simulation from that design, run for this draft, to give a first empirical check of whether the transport mechanism behaves as the theory of Section~\ref{sec:theory} suggests, rather than leaving the comparison entirely unexecuted. This is not a substitute for the full study; sample sizes and replication counts are modest and the design is simplified in ways we flag below.

\paragraph{Design.} A single focal cohort $g$ is treated at $t=12$ in a panel of $T=34$ periods. Untreated potential outcomes follow the interactive-fixed-effects model $Y_{it}(\infty) = \alpha_i + \delta_t + \lambda_i^\top f_t + \varepsilon_{it}$ with $F=2$ latent factors following independent random walks, unit and cohort fixed effects, and $\varepsilon_{it} \sim N(0, 0.5^2)$. The donor pool consists of $18$ never-treated units and four blocks of $6$ temporarily eligible donors that exit the risk set at $t \in \{16,20,24,28\}$ respectively, so the risk set $\mathcal{D}_{g,h}$ shrinks four times over the evaluation window $h=0,\dots,14$. The true treatment-effect path is the saturating form $\tau_h = 2(1-e^{-0.25h})$. We compare: independent horizon-specific SCM ($\rho_h=0$); RT-SC-DiD with a fixed $\rho=3$; fixed-final-horizon SCM (donors restricted to those surviving through $h=14$ from the start); naive deletion-and-renormalization (weights fit once at $h=0$ and thereafter only renormalized via \eqref{eq:renorm}, with a uniform fallback when surviving mass is numerically zero, and never re-optimized); and a not-yet-treated simple-mean comparator with a DiD baseline correction, in the spirit of \citet{callaway2021difference} but without covariate adjustment or synthetic weighting. All synthetic weights are fit by ridge-penalized, simplex-constrained least squares on an 8-period pre-treatment window ($\lambda = 10^{-3}$); the transport kernel uses $\kappa=1$ and the zero-mass fallback of Remark~\ref{rem:zero-mass} with $\epsilon_0 = 10^{-6}$. Results are averaged over $80$ Monte Carlo replications, each redrawing $\alpha_i$, $\delta_t$, $f_t$, and $\varepsilon_{it}$ (donor and focal-cohort loadings $\lambda_i$ are redrawn once per replication as well; the focal cohort's loading is held at a value close to the donor pool's central range across replications). Code for this simulation is included with the source of this paper.

\paragraph{Results.} Table~\ref{tab:sim-results} reports bias and RMSE by event time $h$ for the five methods, at every second horizon. Figure~\ref{fig:sim-main} plots the full bias and RMSE paths. RT-SC-DiD attains the lowest or near-lowest RMSE at nearly every horizon and has the lowest average absolute bias among the four synthetic-control variants; the advantage over independent horizon-specific SCM is modest but consistent, and grows somewhat at longer horizons where more donor exits have accumulated. Fixed-final-horizon SCM is dominated by RT-SC-DiD and independent SCM at every reported horizon in this design, consistent with the motivating claim that restricting to the smallest, longest-surviving donor pool discards useful short-horizon information. Naive deletion-and-renormalization has RMSE comparable to or worse than the other SCM variants at short horizons and visibly worse at long horizons ($h=12,14$), consistent with Proposition~\ref{prop:naive}: as donor exits accumulate, uncorrected renormalization increasingly misallocates weight relative to the more targeted transport of \eqref{eq:transport}--\eqref{eq:kernel-fallback}. The not-yet-treated simple-mean comparator has the largest bias at most horizons, as expected since it does not perform any pre-treatment balancing, though its RMSE is occasionally competitive because it also carries no weight-estimation variance.

\begin{table}[htbp]
\centering
\small
\begin{tabular}{r|c|cc|cc|cc|cc|cc}
\toprule
& & \multicolumn{2}{c|}{Indep.\ ($\rho=0$)} & \multicolumn{2}{c|}{RT-SC-DiD ($\rho=3$)} & \multicolumn{2}{c|}{Fixed-horizon} & \multicolumn{2}{c|}{Naive renorm.} & \multicolumn{2}{c}{NYT mean} \\
$h$ & $\tau_h$ & Bias & RMSE & Bias & RMSE & Bias & RMSE & Bias & RMSE & Bias & RMSE \\
\midrule
0  & 0.00 & $-$0.01 & 0.36 & $-$0.01 & 0.36 & 0.01    & 0.43 & $-$0.01 & 0.36 & $-$0.02 & 0.39 \\
2  & 0.79 & 0.01    & 0.48 & 0.01    & 0.48 & $-$0.01 & 0.53 & 0.01    & 0.48 & $-$0.05 & 0.47 \\
4  & 1.26 & $-$0.03 & 0.49 & $-$0.03 & 0.47 & $-$0.08 & 0.54 & $-$0.07 & 0.53 & $-$0.08 & 0.54 \\
6  & 1.55 & $-$0.08 & 0.66 & $-$0.08 & 0.64 & $-$0.13 & 0.75 & $-$0.12 & 0.69 & $-$0.18 & 0.63 \\
8  & 1.73 & $-$0.05 & 0.77 & $-$0.04 & 0.73 & $-$0.15 & 0.79 & $-$0.12 & 0.93 & $-$0.11 & 0.75 \\
10 & 1.84 & $-$0.12 & 0.78 & $-$0.11 & 0.74 & $-$0.17 & 0.79 & $-$0.19 & 0.95 & $-$0.17 & 0.86 \\
12 & 1.90 & $-$0.24 & 0.98 & $-$0.09 & 0.94 & $-$0.24 & 0.98 & $-$0.15 & 1.28 & $-$0.30 & 0.90 \\
14 & 1.94 & $-$0.21 & 1.06 & $-$0.14 & 1.01 & $-$0.21 & 1.06 & $-$0.12 & 1.42 & $-$0.32 & 0.99 \\
\bottomrule
\end{tabular}
\caption{Bias and RMSE by event time $h$, averaged over 80 Monte Carlo replications of the pilot design described in Section~\ref{sec:simulation-illustration}. NYT mean is the not-yet-treated simple-mean comparator with DiD baseline correction.}
\label{tab:sim-results}
\end{table}

\begin{figure}[htbp]
\centering
\includegraphics[width=0.95\textwidth]{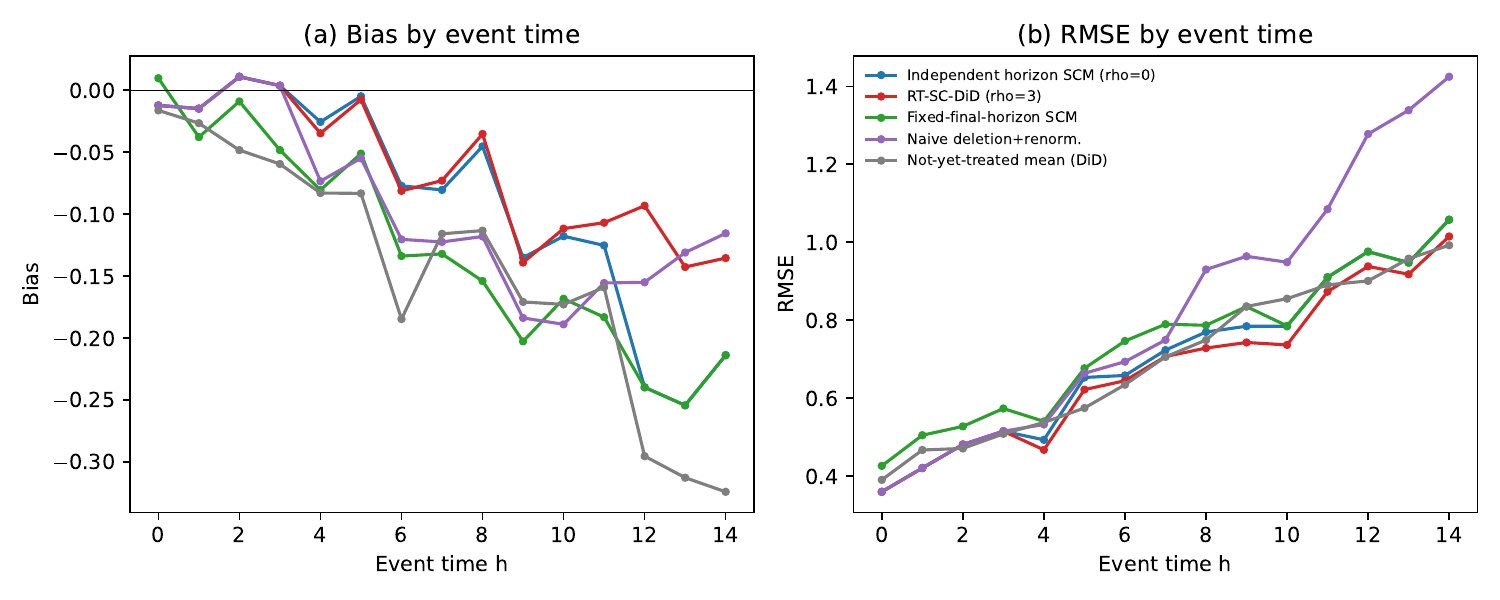}
\caption{Bias (a) and RMSE (b) by event time $h$ for the five methods in Table~\ref{tab:sim-results}, full horizon range $h=0,\dots,14$.}
\label{fig:sim-main}
\end{figure}

\paragraph{Transport-strength sensitivity.} Figure~\ref{fig:sim-rho} and Table~\ref{tab:sim-rho} repeat the RT-SC-DiD arm alone across $\rho \in \{0, 0.3, 3, 50\}$, spanning independent horizon-specific fitting ($\rho=0$) through strong anchoring to the transported reference ($\rho=50$; this is not the same as naive renormalization, since the horizon-$h$ problem is still re-optimized subject to the penalty rather than skipped entirely, and since our $\kappa=1$ in this design keeps the transported reference itself similarity-based rather than close to uniform renormalization). Consistent with the bias--variance motivation for transport regularization set out in Section~\ref{sec:theory} -- though not a result that theory proves -- average absolute bias falls monotonically as $\rho$ increases (heavier anchoring to a stable reference reduces erratic reoptimization), while RMSE is minimized at an intermediate value ($\rho=0.3$ in this design) and rises again at $\rho=50$, where the estimator is anchored so heavily to the transported reference that it under-responds to genuine pre-treatment imbalance. $\rho=0$ and very large $\rho$ are both dominated by an intermediate choice on RMSE in this design, and the specific minimizing $\rho$ is design-dependent, which is exactly why Section~\ref{sec:tuning} proposes selecting it by donor-only placebo cross-validation rather than fixing it a priori. We caution that this is one design with $40$ replications per $\rho$ value on a coarse four-point grid; it illustrates the mechanism rather than establishing the optimal $\rho$ in any generality.

\begin{table}[htbp]
\centering
\small
\begin{tabular}{lcccc}
\toprule
& $\rho=0$ & $\rho=0.3$ & $\rho=3$ & $\rho=50$ \\
\midrule
Mean $|$bias$|$ across $h$ & 0.079 & 0.053 & 0.050 & 0.046 \\
Mean RMSE across $h$ & 0.639 & \textbf{0.567} & 0.617 & 0.621 \\
\bottomrule
\end{tabular}
\caption{RT-SC-DiD sensitivity to the transport strength $\rho$, averaged over event time $h=0,\dots,14$ and 40 Monte Carlo replications per $\rho$ value.}
\label{tab:sim-rho}
\end{table}

\begin{figure}[htbp]
\centering
\includegraphics[width=0.95\textwidth]{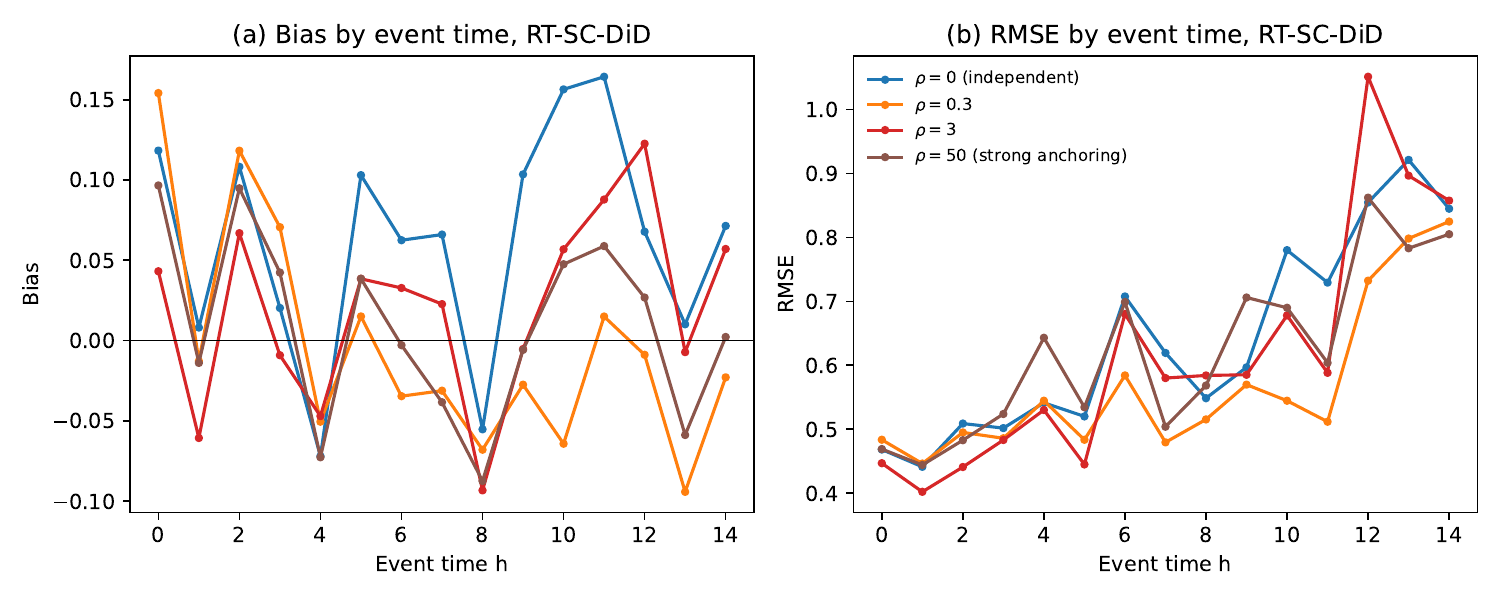}
\caption{RT-SC-DiD bias (a) and RMSE (b) by event time $h$, across transport strength $\rho \in \{0, 0.3, 3, 50\}$.}
\label{fig:sim-rho}
\end{figure}

\paragraph{Limitations of this pilot.} This simulation is small on every dimension listed in Section~\ref{sec:simulation}.2: one never-treated share, one rollout speed, one attrition pattern, one factor dimension, one pre-treatment length, random (not covariate- or loading-based) treatment timing, one treatment-effect path, and no spillovers. It uses $80$ replications for the main method comparison (and only $40$ per value for the transport-strength sensitivity analysis), an order of magnitude below what we would consider adequate for reporting Monte Carlo standard errors on the bias and RMSE estimates themselves. We report it because a single honest illustrative check, with its limitations stated, is more informative than no check at all, and because it directly addresses whether the transport mechanism does anything at all in a setting built to match the problem it targets -- but the full design of Section~\ref{sec:simulation}.1--\ref{sec:simulation}.3 remains unexecuted, and we do not treat this pilot as validating the method beyond this one design.

\section{Empirical application}
\label{sec:application}

An informative application should combine: many treatment-adoption cohorts; a long pre-treatment panel; relatively few never-treated units, so that not-yet-treated donors carry real informational weight; several later-treated units usable as temporary donors; plausible heterogeneity in untreated trends across units; well-recorded treatment dates; and limited anticipation and spillovers. Candidate settings include environmental regulations adopted by jurisdictions at staggered dates, phased rollouts of hospital technologies, staggered municipal transportation or educational policies, firm-level staggered platform interventions, regional labor-market regulations, and phased eligibility or reimbursement rules.

Rather than reporting only that RT-SC-DiD changes the point estimate relative to conventional approaches, we recommend that an application include a decomposition addressing: (1) how much additional predictive information is gained from temporary donors relative to a fixed-final-horizon pool; (2) how much weight mass is lost as donors become treated, horizon by horizon; (3) how the transport penalty redistributes that mass, and to which donors; (4) whether conventional horizon-by-horizon synthetic control produces visible jumps that RT-SC-DiD does not; and (5) how sensitive the substantive conclusion is to restricting entirely to donors untreated through the final horizon.

\section{Falsification and sensitivity analysis}
\label{sec:falsification}

\paragraph{Pre-treatment placebo effects.} Estimate $\widehat\tau_{g,h}$ for $h < 0$ using pseudo-adoption dates; large placebo effects indicate inadequate balance or unstable transport.

\paragraph{Leave-one-donor-out.} Recompute $\widehat\tau_{g,h}^{(-j)}$ omitting donor $j$ and report $L_{g,h} = \max_j |\widehat\tau_{g,h}^{(-j)} - \widehat\tau_{g,h}|$.

\paragraph{Treatment-date perturbation.} Shift adoption dates within plausible administrative uncertainty and recompute.

\paragraph{Anticipation windows.} Exclude donors within $a$ periods of their own adoption date, $\mathcal{D}_{g,h}^{(a)} = \{i : G_i > g+h+a\}$, to probe sensitivity to mild anticipation.

\paragraph{Fixed-pool comparison.} Compare against a conservative pool restricted to units untreated through $g+H$.

\paragraph{Transport-strength path.} Plot estimates over a grid of $\rho \in \{0, \rho_1, \rho_2, \dots, \infty\}$, from independent horizon fits ($\rho = 0$) to maximal continuity ($\rho \to \infty$), to visualize the bias--stability trade-off directly.

\section{Discussion and scope}
\label{sec:discussion}

We want to be explicit about what this paper does and does not claim. It does not claim to be the first synthetic-control estimator for staggered adoption, the first method to use not-yet-treated units as donors, or a fundamentally new causal framework distinct from Callaway--Sant'Anna-style group-time effects; nor does it claim that synthetic balancing removes the need for identifying assumptions, that later-treated units are automatically valid controls, or that treatment timing among donors is ignorable simply because they have not yet been treated -- Assumption~\ref{ass:attrition} is required precisely because this is not automatic. The defensible claim is narrower: this paper proposes a specific way to construct a coherent \emph{sequence} of synthetic counterfactuals when donor eligibility changes mechanically across post-treatment horizons, together with a transport mechanism whose properties (Propositions~\ref{prop:reopt} and~\ref{prop:naive}, and Theorem~\ref{thm:propagation}) are stated and proved, and a diagnostic toolkit for assessing when the resulting estimates rest on adequate donor support. Two limitations follow directly from this scope. First, the benefit of RT-SC-DiD relative to a fixed-final-horizon design is expected to be largest when no large never-treated group exists and later-treated units carry substantial predictive information at short horizons; when a large stable never-treated group is available, the gains from horizon-varying donor pools are likely to be modest, and simpler estimators may be preferred on grounds of transparency. Second, the method's validity continues to rest on Assumption~\ref{ass:attrition}, which is not testable from the observed data alone; the falsification checks in Section~\ref{sec:falsification} provide indirect evidence but not proof.

\section{Conclusion}
\label{sec:conclusion}

Staggered treatment-adoption designs create a donor pool that is not fixed but contracts mechanically, and informatively, as units enter treatment. Existing staggered synthetic control and difference-in-differences methods already recognize that eligible controls for a given cohort depend on the evaluation horizon, but common practice resolves this either by discarding temporarily eligible donors or by re-estimating weights independently at each horizon, at the cost of either lost information or reweighting-induced instability. We have proposed RT-SC-DiD, which fits horizon-specific synthetic weights on the full currently valid donor risk set while regularizing toward a similarity-weighted transport of the previous horizon's weights, combined with a DiD-style baseline correction and an optional partially pooled objective across cohorts; we have been explicit that this baseline correction is not yet a fully optimized SDiD time-weight structure, and have named the estimator accordingly, leaving genuine SDiD time weighting as a stated extension. We have separated identification (an oracle statement), estimation consistency (an open item requiring further work), and the finite-sample bias introduced by the transport penalty itself, rather than treating consistent estimation as given. We have proved that $\ell_1$ transport distance directly bounds the same-period reoptimization distortion, with the corresponding dimension-adjusted bound for the $\ell_2$ distance used inside the optimization objective; corrected and completed the characterization of the bias from naive deletion-and-renormalization in terms of exited weight mass and exit-survivor loading discrepancy; and derived a recursive bound on how estimation and transport-approximation error propagate across horizons under an explicitly stated (not derived) Lipschitz regularity condition on the transport map, showing that the mechanism trades reduced per-horizon variance for genuine, bounded-but-nonzero error propagation rather than being a costless stabilizer. We have also laid out a donor-only placebo procedure for tuning-parameter selection, a set of donor-support diagnostics including both $\ell_1$ and $\ell_2$ transport measures, and a fuller Monte Carlo and empirical-application protocol -- with a decomposed counterfactual-discontinuity metric -- for evaluating the method against existing staggered DiD, synthetic control, synthetic DiD, staggered-SCM inference, and matrix-completion estimators. A small illustrative pilot simulation, reported in Section~\ref{sec:simulation-illustration}, shows RT-SC-DiD attaining the lowest or near-lowest RMSE among synthetic-control variants at most horizons and an intermediate transport strength minimizing RMSE, consistent with the bias--variance motivation of the method rather than a result our theory proves, though this single 80-replication design does not substitute for the fuller study we propose, and no empirical application is yet reported. We regard the central contribution as a formal treatment of the informative contraction of the admissible donor set across event-time horizons -- a problem that existing methods implicitly confront but do not, to our knowledge, address with a proved distortion bound and an accompanying diagnostic toolkit -- rather than as a new causal-inference framework in its own right or as a completed synthetic difference-in-differences estimator.

\bibliographystyle{apalike}

\end{document}